\newif\ifshort
\shorttrue
\newif\ifappendix
\appendixtrue

\documentclass[11pt]{article}

\usepackage[a4paper,includefoot,margin=3cm]{geometry}
\usepackage{amsmath, amssymb, amsthm}
\usepackage{nicefrac}
\usepackage[mathic]{mathtools}
\usepackage{thm-restate}
\usepackage{graphicx,scalerel}
\usepackage[dvipsnames]{xcolor}
\usepackage{subcaption}
\usepackage{multirow}
\usepackage{authblk}

\usepackage[textsize=footnotesize,backgroundcolor=green!40!white,linecolor=black!60,obeyFinal]{todonotes}
\usepackage[T1]{fontenc}

\usepackage{dsfont, paralist, microtype, enumerate}
\usepackage{booktabs, tabularx}

\usepackage[pdfdisplaydoctitle,menucolor=orange!40!black,filecolor=magenta!40!black,urlcolor=blue!40!black,linkcolor=red!40!black,citecolor=green!40!black,ocgcolorlinks]{hyperref} %

\newcommand{\ExternalLink}{%
    \tikz[color=magenta, x=1.2ex, y=1.2ex, baseline=-0.05ex]{%
        \begin{scope}[x=1ex, y=1ex]
            \clip (-0.1,-0.1) 
                --++ (-0, 1.2) 
                --++ (0.6, 0) 
                --++ (0, -0.6) 
                --++ (0.6, 0) 
                --++ (0, -1);
            \path[draw, 
                line width = 0.5, 
                rounded corners=0.5] 
                (0,0) rectangle (1,1);
        \end{scope}
        \path[draw, line width = 0.5] (0.5, 0.5) 
            -- (1, 1);
        \path[draw, line width = 0.5] (0.6, 1) 
            -- (1, 1) -- (1, 0.6);
        }
}

\DeclareUnicodeCharacter{1EF3}{\`y}
\usepackage[
	backend=biber,
	isbn=false,
	url=false,
	doi=true,
	backref=true,
	hyperref=true,
	natbib=true,
	maxcitenames=3,
	maxbibnames=99,
	sortcites]{biblatex}
\renewbibmacro{in:}{\ifentrytype{article}{}{\printtext{\bibstring{in}\intitlepunct}}}
\newbibmacro{doilink}[1]{%
	\iffieldundef{doi}{%
		\iffieldundef{url}{%
			#1%
        }{%
		\href{\thefield{url}}{$\ExternalLink$}%
		}%
	}{%
		\href{https://doi.org/\thefield{doi}}{$\ExternalLink$}%
	}%
}
\DeclareFieldFormat*{doi}{\usebibmacro{doilink}{#1}}

\bibliography{strings-long,references}

\usepackage[capitalise]{cleveref}
\crefname{alg}{Algorithm}{Algorithms}
\Crefname{alg}{Algorithm}{Algorithms}
\crefname{axiom}{}{}
\creflabelformat{axiom}{(#2I#1#3)}

\usepackage{url}
\usepackage{paralist}

\usepackage{tikz}
\usetikzlibrary{arrows, arrows.meta, calc, decorations.pathmorphing, decorations.pathreplacing, calligraphy, backgrounds, math, shapes, shapes.geometric, spy, positioning, quotes}

\tikzstyle{path} = [color=black!10,line cap=round, line join=round, line width=12pt]
\tikzset{
	vertex/.style={
		circle,
		draw=black,
		solid,
		thick,
		fill=white,
		minimum size=2mm,
		inner sep=0,
	},
	edge/.style={
		draw=black,
		semithick,
	},
	arc/.style={
		edge,
		->,
		>=stealth',
	},
	colornode/.style = {
		circle,
		draw=#1!70!black,
		very thick,
		fill=#1
	},
	smallcolornode/.style = {
		colornode=#1,
		thick,
		scale=0.65
	}
}

\makeatletter

\def\@maketitle{%
  \newpage
  \null
  \vskip 2em%
  \begin{center}%
  \let \footnote \thanks
    {\Large\bf \@title \par}%
    \vskip 1.5em%
    {\large
      \lineskip .5em%
      \begin{tabular}[t]{c}%
        \@author
      \end{tabular}\par}%
    \vskip 1em%
    {\large \@date}%
  \end{center}%
  \par
  \vskip 1.5em}

\def\url@leostyle{%
  \@ifundefined{selectfont}{\def\UrlFont{\sf}}{\def\UrlFont{\small\ttfamily}}}
\makeatother
\usepackage{thmtools}
\theoremstyle{plain}
\newtheorem{theorem}{Theorem}
\crefname{theorem}{Theorem}{Theorems}
\Crefname{theorem}{Thm{.}}{Thms{.}}
\newtheorem{lemma}[theorem]{Lemma}

\newtheorem{proposition}[theorem]{Proposition}

\newtheorem{corollary}[theorem]{Corollary}
\newtheorem{observation}[theorem]{Observation}
\Crefname{observation}{Obs{.}}{Obs{.}}
\crefname{observation}{Observation}{Observations}

\theoremstyle{definition}
\newtheorem{definition}[theorem]{Definition}
\newtheorem{alg}{Algorithm}

\theoremstyle{remark}

\declaretheorem[style=definition,name=Construction,qed=$\diamond$]{construction}

\newcommand{\prob}[1]{\textnormal{\textsc{#1}}}

\newcommand{\probDef}[3]{
	\begin{center}
	\begin{minipage}{\columnwidth-6em}
		\noindent
		\prob{#1}
		\vspace{5pt}\\
		\setlength{\tabcolsep}{3pt}
		\begin{tabularx}{\textwidth}{@{}lX@{}}
			\textbf{Input:}     & #2 \\
			\textbf{Question:}  & #3
		\end{tabularx}
	\end{minipage}
	\end{center}
}
\DeclarePairedDelimiterX{\abs}[1]{\lvert}{\rvert}{#1}
\DeclarePairedDelimiterX{\norm}[1]{\lVert}{\rVert}{#1}
\DeclarePairedDelimiterX{\ceil}[1]{\lceil}{\rceil}{#1}

\newcommand{\NN}{\ensuremath{\mathds{N}}}
\newcommand{\Nzero}{\ensuremath{\NN_0}}
\newcommand{\ZZ}{\ensuremath{\mathds{Z}}}

\newcommand{\FFF}{\ensuremath{\mathcal{F}}}

\newcommand{\NNN}{\ensuremath{\mathcal{N}}}
\newcommand{\PPP}{\ensuremath{\mathcal{P}}}
\newcommand{\QQQ}{\ensuremath{\mathcal{Q}}}
\newcommand{\RRR}{\ensuremath{\mathcal{R}}}
\newcommand{\SSS}{\ensuremath{\mathcal{S}}}
\newcommand{\TTT}{\ensuremath{\mathcal{T}}}
\newcommand{\WWW}{\ensuremath{\mathcal{W}}}
\newcommand{\YYY}{\ensuremath{\mathcal{Y}}}

\newcommand{\tildy}[1]{\widetilde{#1}}
\newcommand{\hatty}[1]{\widehat{#1}}

\newcommand{\FPT}{\ensuremath{\mathrm{FPT}}}
\newcommand{\NP}{\ensuremath{\mathrm{NP}}}

\newcommand{\bigO}{\ensuremath{\mathcal{O}}}
\newcommand{\yes}{\textnormal{\texttt{yes}}}
\newcommand{\no}{\textnormal{\texttt{no}}}

\newcommand{\true}{\texttt{true}}
\newcommand{\false}{\texttt{false}}

\DeclareMathOperator{\dist}{dist}

\DeclareMathOperator{\rep}{rep}
\DeclareMathOperator{\orep}{orep}
\DeclareMathOperator{\prep}{prep}

\newcommand{\repr}{\ensuremath{\subseteq_{\rep}}}
\newcommand{\orepr}{\ensuremath{\subseteq_{\orep}}}
\newcommand{\prepr}{\ensuremath{\subseteq_{\prep}}}

\newcommand{\representative}[1][$q$]{#1-re\-pre\-sen\-ta\-tive}
\newcommand{\orepresentative}[1][$r$]{ordered \representative[#1]}
\newcommand{\prepresentative}[1]{partial #1-re\-pre\-sen\-ta\-tive}

\newcommand{\xtoy}[2]{\ensuremath{[ #1, #2 ]}} %
\newcommand{\oneto}[1]{\ensuremath{[ #1 ]}} %
\newcommand{\nullto}[1]{\xtoy{0}{#1}} %

\newcommand{\ceq}{\ensuremath{\coloneqq}}

\newcommand{\abX}[3]{\ensuremath{#1}\nobreakdashes-\ensuremath{#2}~#3}
\newcommand{\walk}{walk}
\newcommand{\walks}{walks}
\newcommand{\abpath}[2]{\abX{#1}{#2}{path}}
\newcommand{\abpaths}[2]{\abX{#1}{#2}{paths}}
\newcommand{\stpath}{\abpath{s}{t}}
\newcommand{\stpaths}{\abpaths{s}{t}}
\newcommand{\abwalk}[2]{\abX{#1}{#2}{\walk}}
\newcommand{\abwalks}[2]{\abX{#1}{#2}{\walks}}
\newcommand{\stwalk}{\abwalk{s}{t}}
\newcommand{\stwalks}{\abwalks{s}{t}}
\newcommand{\subp}{subpath}
\newcommand{\subw}{sub\walk}

\newcommand{\pref}{rainbow}
\newcommand{\prefness}{rainbowness}
\newcommand{\ppref}{locally rainbow}
\newcommand{\pprefness}{local rainbowness}
\newcommand{\rppref}{\ensuremath{r}-\pref}
\newcommand{\rpprefness}{\ensuremath{r}-\prefness}

\newcommand{\ppRef}{Locally Rainbow}
\newcommand{\fp}{\prob{\ppRef{} Path}}
\newcommand{\fw}{\prob{\ppRef{} Walk}}

\newcommand{\col}{\ensuremath{c}}
\newcommand{\colors}{\ensuremath{C}}
\newcommand{\ncols}{\ensuremath{\abs{\colors}}}

\newcommand{\fpt}{fixed-parameter tractable}

\usepackage{xargs}
\newcommandx{\set}[2][1=1]{\ensuremath{\{#1,\ldots,#2\}}} %
\newcommandx{\tlog}[3][1=,3=]{\log_{#1}^{#3}(#2)}
\newcommandx{\ith}[2][1=th]{#2\nobreakdash-#1}

\usepackage{etoolbox}
\newcommand{\appsymb}{$\bigstar$}

\ifshort{}\ifappendix{}
\newcommand{\appref}[1]{\hyperref[proof:#1]{\appsymb}}
\else{}
\newcommand{\appref}[1]{\appsymb}
\fi{}
\else{}
\newcommand{\appref}[1]{}
\fi{}

\newcommand{\appendixsection}[1]{%
	\ifshort{}%
	\gappto{\appendixProofText}{\section{Additional Material for Section~\ref{#1}}\label{app:#1}}%
	\fi{}%
}

\newcommand{\toappendix}[1]{%
\ifshort{}%
  \gappto{\appendixProofText}
  {{
    #1
  }}
\else{}#1\fi{}%
}

\newcommand{\appendixproof}[2]{%
\ifshort{}
  \gappto{\appendixProofText}
  {
    \subsection[Missing Proof]{Proof of \cref{#1}}\label{proof:#1}
    #2
  }%
\else{}
#2
\fi{}
}

\title{Locally Rainbow Paths}

\author{Till Fluschnik\thanks{Supported by the DFG, project AFFA (BR~5207/1).}$^,$}
\affil{\normalsize Institut für Informatik, TU Clausthal, Germany}

\author{Leon Kellerhals}

\author{Malte Renken}
\affil{%
	\normalsize Technische Universität Berlin, Algorithmics and Computational Complexity, Germany

	\bigskip
	\small\texttt{till.fluschnik@tu-clausthal.de,\,\{leon.kellerhals,\,m.renken\}@tu-berlin.de}
}
\date{}

\begin{document}

\maketitle

\begin{abstract}
	We introduce the algorithmic problem of finding a \emph{locally rainbow} path of length $\ell$ connecting two distinguished vertices $s$ and $t$ in a vertex-colored directed graph.
	Herein, a path is locally rainbow if between any two visits of equally colored vertices, the path traverses consecutively at least $r$ differently colored vertices.
	This problem generalizes the well-known problem of finding a rainbow path.
	It finds natural applications whenever there are different types of resources that must be protected from overuse, such as crop sequence optimization or production process scheduling.
	We show that the problem is computationally intractable even if $r=2$ or if one looks for a locally rainbow among the shortest paths.
	On the positive side, if one looks for a path that takes only a short detour (i.e., it is slightly longer than the shortest path) and if $r$ is small, the problem can be solved efficiently.
	Indeed, the running time of the respective algorithm is near-optimal unless the ETH fails.
\end{abstract}

\section{Introduction}

Many graph connectivity problems are studied with additional constraints to make them applicable to real-world problems.
Typical constraints include forbidden pairs of vertices or edges in the solution or
--- if the graph is colored --- requiring that the solutions are rainbow (no two elements in the solution have the same color) or properly colored (no two adjacent elements have the same color).
Examples for such constraints on problems can be found for spanning trees \cite{broersma1997spanning,darmann2011disjunctive}, Steiner trees \cite{una2016steiner,ferone2022rainbow,haldorsson2018spanning}, but most notably for paths \cite{alon1995color,agrawal2020conflictfree,bhattacharya2010search,bentert2023paths}.

For paths, the properly edge-colored variant forbids two equally colored edges to appear subsequently in the path.
What, to the best of our knowledge, has not been considered yet, is any model that forbids a visited color for the next, say $r$, subsequent vertices of the path.
For example, this allows the modeling of protecting certain types of resources from overuse.
This for example is relevant for crop sequence optimization: here, different colors model different types of crops which, depending on the season, have different impacts on soil health \cite{dury2012cropping,turchetta2022crop,benini2023crop}.
Other applications include holiday trip planning (different colors modeling different types of leisure activities),
production process scheduling (different colors modeling different workers or machines).

More concretely, given a vertex-colored graph, we propose the concept of \emph{\ppref{}} paths, in which every subpath of bounded length is required to carry pairwise distinct colors.
Formally, 
a path or walk~$W=(v_0,v_1,\dots,v_q)$ in~$G$ is \emph{\rppref{}}
if for every~$i\in\nullto{q-r}$,
the vertices~$v_i,v_{i+1},\dots,v_{i+r}$ have pairwise different color
(see \cref{fig:ppref-path} for an illustration).
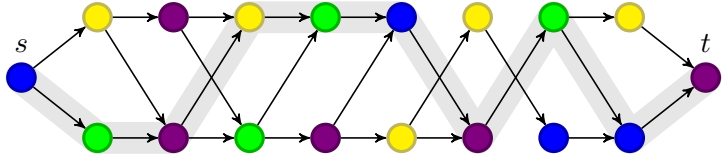
\begin{figure}[tb]
\centering
\ifdefined\externaltikz\tikzsetnextfilename{fishing-tours-example}\fi%
\begin{tikzpicture}[yscale=.8]
		\tikzstyle{anode}=[colornode=blue]
		\tikzstyle{bnode}=[colornode=violet]
		\tikzstyle{cnode}=[colornode=green]
		\tikzstyle{dnode}=[colornode=yellow]
		\tikzstyle{enode}=[colornode=gray]
		\node[anode, label=$s$] at (-4.5,0) (s) {};
		\node[bnode, label=$t$] at(4.5,0) (t) {};
		
		\node[cnode] at (-3.5,-1) (p11) {};
		\node[bnode] at (-2.5,-1) (p12) {};
		\node[cnode] at (-1.5,-1) (p13) {};
		\node[bnode] at (-.5,-1) (p14) {};
		\node[dnode] at (.5,-1) (p15) {};
		\node[bnode] at (1.5,-1) (p16) {};
		\node[anode] at (2.5,-1) (p17) {};
		\node[anode] at (3.5,-1) (p18) {};
		
		\node[dnode] at (-3.5,1) (p21) {};
		\node[bnode] at (-2.5,1) (p22) {};
		\node[dnode] at (-1.5,1) (p23) {};
		\node[cnode] at (-.5,1) (p24) {};
		\node[anode] at (.5,1) (p25) {};
		\node[dnode] at (1.5,1) (p26) {};
		\node[cnode] at (2.5,1) (p27) {};
		\node[dnode] at (3.5,1) (p28) {};

		\tikzstyle{myarc}=[->,>=stealth',semithick]
			\draw[myarc] (s) to (p11);
			\draw[myarc] (s) to (p21);
			\draw[myarc] (p11) to (p12);
			\draw[myarc] (p21) to (p12);
			\draw[myarc] (p21) to (p22);
			\draw[myarc] (p12) to (p13);
			\draw[myarc] (p12) to (p23);
			\draw[myarc] (p22) to (p13);
			\draw[myarc] (p22) to (p23);
			\draw[myarc] (p13) to (p14);
			\draw[myarc] (p13) to (p24);
			\draw[myarc] (p23) to (p24);
			\draw[myarc] (p14) to (p15);
			\draw[myarc] (p14) to (p25);
			\draw[myarc] (p24) to (p25);
			\draw[myarc] (p15) to (p16);
			\draw[myarc] (p15) to (p26);
			\draw[myarc] (p25) to (p16);
			\draw[myarc] (p16) to (p27);
			\draw[myarc] (p26) to (p17);
			\draw[myarc] (p17) to (p18);
			\draw[myarc] (p27) to (p18);
			\draw[myarc] (p27) to (p28);
			\draw[myarc] (p18) to (t);
			\draw[myarc] (p28) to (t);

		\begin{scope}[on background layer]
			\draw[path] (s.center) -- (p11.center) -- (p12.center) -- (p23.center) -- (p24.center) -- (p25.center) -- (p16.center) -- (p27.center) -- (p18.center) -- (t.center);
		\end{scope}
	\end{tikzpicture}
	\caption{
		A digraph whose vertices are colored with four colors,
		with a shortest $2$-\pref{} (but not $3$-\pref{}) \stpath{}.
	}
	\label{fig:ppref-path}	
\end{figure}
We arrive at the following problem description:

\probDef{\fp{}}
{A digraph~$G$, a vertex-coloring~$\col \colon V(G) \to \colors$, two distinct vertices~$s, t \in V(G)$, and two integers~$r,\ell\in\Nzero$.}
{Is there an \rppref{} \stpath{} of length at most~$\ell$ in~$G$?}

We also consider \fw{}, where we look for \stwalks{} with the same constraints.
\fp{} becomes the aforementioned
problem of finding a rainbow path
when $r$~equals the path's length;
on the other extreme, 
if~$r=1$,
then the problem coincides with finding a properly colored \stpath{}.

\paragraph{Our contributions.}
We study the parameterized complexity of \fw{} and \fp{},
with a focus on the \emph{locality} parameter~$r$.
We show that the path variant is \NP-hard for any fixed value of $r \ge 2$ (\cref{thm:ufp}).
In contrast, we are able to design an algorithm with running time~$2^{\bigO(r \log r)} \cdot n^{\bigO(1)}$  for the walk variant (\cref{thm:dft-fpt-r}), with $n$ being the number of vertices.
This result is achieved by developing an ordered version of the \emph{representative families} technique.
We prove this result to be optimal in the sense that no $2^{o(r \log r)} \cdot n^{\bigO(1)}$-time algorithm is possible if the ETH holds (\cref{thm:r-eth}).

Note that an \rppref{} \stwalk{} of length~$\ell$ must always be a path when $\ell \leq \dist(s, t) + r$.
Thus, our algorithm for \fw{} also applies to the path variant when the \emph{detour length} $k \coloneqq \ell - \dist(s, t)$ is small.
Motivated by this observation and the result of~\citet{bezakova2019detours} that finding \stpaths{} of detour length~$k$ is fixed-parameter tractable for the parameter~$k$,
we also investigate this parameter.
While both of our problem variants remain \NP-hard even when~$k = 0$ (\cref{thm:r-eth}),
we are able to give a fixed-parameter tractable algorithm for the combined parameter~$k + r$ (\cref{thm:dfp-above}).

We mention in passing that our results also hold when coloring the edges instead of vertices (and adapting \pprefness{} accordingly).
Furthermore, our (nontrivial) algorithmic results also hold when looking for paths of length \emph{exactly} $\ell$.
\ifshort{}
Proofs of results marked with \appsymb{}
\ifappendix{}are deferred to the appendix.%
\else{}are deferred to the paper's full version.\fi{}\fi{}
 
\paragraph{Related work.}
Finding a rainbow path is known to be \NP-hard \cite{chen2011determining} and fixed-parameter tractable with respect to the number of colors \cite{kowalik2016rainbow,uchizawa2013rainbow}.
While 
finding a properly colored path
is trivially linear-time solvable,
it is less obvious that 
this is also solvable in that time 
in an (undirected) \emph{edge}-colored graph.
This was shown by \citet{szeider2003transitions}.

The field of finding paths of detour length exactly or at least~$k$ is rather active, with the former being easier to tackle than the latter.
\citet{bezakova2019detours} prove both variants to be fixed-parameter tractable, however,
for the latter variant only on undirected graphs.
While there has been some progress on directed graphs, most recently by \citet{jacob2023detours}, 
it is open whether finding a path with detour length at least~one is polynomial-time solvable.

Another closely related and more applied area is that of finding resource-constrained paths.
Here, the graph carries arc (or vertex) weights and the desired \stpath{} must not accumulate more than a given threshold of that weight.
The problem is known to be \NP-hard \cite{HZ80} and 
studied in many variations \cite{ford2022backtracking,Irnich2005,PG13}.
A variation close to our setting introduces so-called \emph{replenishment arcs}, at which one may ``drop off'' the weight accumulated so far \cite{smith2012replenishment}.
This setting is relevant in airline/train crew scheduling (weight represents duty hours, replenishment arcs correspond to crew overnight rests) and aircraft/train routing (weight represents machine hours, replenishment arcs correspond to maintenance events) and also has ties with electric vehicle routing problems (weight represents battery discharge, replenishment arcs correspond to charging events) \cite{adler2016electric,zundorf2014ev}.
Our \prefness{} constraint is similar in that it ``replenishes'' any colors that were visited more than~$r$ steps ago.

\section{Preliminaries}

We denote by~$\ZZ$, $\Nzero$, and $\NN$ the set of all, the non-negative, and the positive integers, respectively.
For~$n, m \in \ZZ$ we denote by~$\xtoy{n}{m} \coloneqq \{ i \in \ZZ \mid n \le i \le m \}$ the set of integers between~$n$ and~$m$ and define~$\oneto{n} \coloneqq \xtoy{1}{n}$.
We denote by~$e \approx 2.718$ Euler's number and by~$\omega < 2.373$ the matrix multiplication constant~\cite{matrixmultiplication}.

Let~$\sigma \coloneqq (a_1, \dots, a_n)$ be a sequence.
We denote by~$\abs{\sigma} \coloneqq n$ its \emph{length}, i.e., the number of elements in~$\sigma$, and also call~$\sigma$ an~\emph{$n$-sequence}.
We write~$x \in \sigma$ if~$x = a_i$ for some~$i \in \oneto{n}$.
If every element in~$\sigma$ is contained in a set~$U$, then we say that~$\sigma$ is a \emph{sequence on} (or \emph{over})~$U$.
A sequence~$\sigma'$ is a \emph{substring} or \emph{consecutive subsequence} of~$\sigma$ if there are~$i < j \in \oneto{n}$ with~$\sigma' = (a_i, a_{i+1}, \dots, a_j)$.
If~$i=1$ or~$j=n$, then we also say that~$\sigma$ \emph{begins with} or \emph{ends on} $\sigma'$, respectively.
If~$\rho = (b_1, \dots, b_m)$ is a sequence, then we denote by~$\sigma \circ \rho \coloneqq (a_1, \dots, a_n, b_1, \dots, b_m)$ the \emph{concatenation} of~$\sigma$ and~$\rho$.
For sequences~$\sigma_1, \dots, \sigma_n$, we denote by~$\bigcirc_{i=1}^n \sigma_i = \sigma_1 \circ \dots \circ \sigma_n$ their consecutive concatenation.

\paragraph{Graph theory.}

For basic notations on (directed) graph theory see, e.g., \cite{Diestel,BangJensenG09}.
A digraph~$G$ is a tuple~$(V,A)$ with~$A\subseteq V\times V$.
In this work,
all digraphs contain no self-loops,
i.e.,
no arcs from the set~$\{(v,v)\mid v\in V\}$.
For a digraph~$G=(V,A)$ we also denote by~$A(G)$ the arc set~$A$ 
and by~$V(G)$ the vertex set~$V$.
We call a digraph~$G$ symmetric if~$(v,w)\in A(G)\iff (w,v)\in A(G)$.
The symmetrization of the digraph~$G$ is the graph~$(V,A(G)\cup\{(v,w)\mid (w,v)\in A(G)\})$.
For two vertices~$v,w\in V(G)$,
a \abwalk{v}{w} $W=(u_0=v,u_1,\dots,u_q=w)$ (of length~$q$) is a sequence of vertices from~$V$ such that~$(u_{i-1},u_{i})\in A(G)$ for every~$i\in\oneto{q}$.
A \abwalk{v}{w} is a path if all vertices are pairwise different.
A digraph~$G$ is weakly connected if in its symmetrization~$G^*$
it holds true that for any~$(v,w)\in V\times V$ there is an~\abpath{v}{w}.
Throughout, 
unless stated otherwise,
we denote by $n \coloneqq \abs{V(G)}$ and~$m \coloneqq \abs{A(G)}$
and
assume the input digraph~$G$ to be weakly connected (and hence~$n \le m-1$).
For a vertex~$v$,
we denote by~$N^-(v)\ceq \{w\in V(G)\mid (w,v)\in A(G)\}$.

\paragraph{Color sequences and compatibility.}
Let~$G$ be a digraph and let~$c \colon V(G) \to \colors$ be a vertex coloring.
Recall that we call a path or walk~$W=(v_0,v_1,\dots,v_q)$ in~$G$
\emph{\rppref{}}
if for every~$i\in\nullto{q-r}$,
the vertices~$v_i,v_{i+1},\dots,v_{i+r}$ have pairwise different color.
The \emph{color sequence} of~$W$ is~$\sigma \coloneqq (\col(v_0), \dots, \col(v_q))$.
We sometimes also call~$\sigma$ \rppref{} if~$W$ is \rppref{}.
For two \rppref{} sequences $\sigma= (a_1, \dots, a_n)$ and~$\rho= (b_1, \dots b_m)$,
we say that~$\sigma$ is \emph{$r$-compatible} to~$\rho$ if a path or walk with color sequence~$\sigma \circ \rho$ is \rppref{}.
\ifshort{}
	Formally, 
	$\sigma$ is $r$-compatible to~$\rho$ if
	$\{a_{\max(1,n-j+1)},\dots,a_n\}\cap \{b_1,\dots,b_{\min(r-j+1,m)}\}=\emptyset$ for all~$j\in\oneto{r}$.
\fi{}

\paragraph{Parameterized complexity.}
Let~$\Sigma$ be a finite alphabet and~$\Sigma^*=\{x\in\Sigma^n\mid n\in\Nzero\}$.
A parameterized problem~$P$ is a subset~$\{(x,k)\mid x\in\Sigma^*,k\in\Nzero\}\subseteq \Sigma^*\times \Nzero$,
where~$k$ is referred to as the parameter.
A parameterized problem~$P$ is \fpt{} (in \FPT{})
if every instance~$(x,k)$ is solvable in~$f(k)\cdot |x|^{O(1)}$ time,
where~$f$ is some computable function only depending on~$k$.
The Exponential Time Hypothesis (ETH)~\cite{ImpagliazzoP01,ImpagliazzoPZ01} states that there
exists some fixed $\varepsilon>0$ such that 
\textsc{3-Sat} cannot be decided in $2^{\varepsilon\cdot n}\cdot (n+m)^{O(1)}$
time on any input with $n$ variables and~$m$ clauses.
For more details,
see~\citet{bluebook}.

\section{Walks}
\label{ft:sec:walks}
\appendixsection{ft:sec:walks}

In this section we study the parameterized complexity of \fw{} with respect to the parameter~$r$.
Note that all results obtained here also hold for finding shortest \rppref{} paths, i.e., paths with length~$\ell = \dist(s, t)$.
We will see that the problem is fixed-parameter tractable, by providing an~$r^{\bigO(r)}\cdot  n^{\bigO(1)}$-time algorithm.
Indeed, although the length of a walk is not bounded in the input size, we can show that the above running time holds even if we ask whether there exists an \rppref{} \stwalk{} of any length.
Finally, we prove a asymptotically tight running time lower bound based on the Exponential Time Hypothesis (ETH).

\subsection{Fixed-Parameter Tractability}
\label{ssec:FPTwalks}

In this section we show the following.

\begin{theorem}
	\label{thm:dft-fpt-r}
	\fw{} can be solved in~$\bigO((r\cdot e)^{\omega r} \cdot \ell m)$ time, where~$m$ is the number of arcs in the input graph and $\omega$ is the matrix multiplication constant.
\end{theorem}

Note that this does not yet prove fixed-parameter tractability for \fw{} parameterized by~$r$ as the walk may become very long, i.e., $\ell$ may not be bounded polynomially in the input size or by any function in~$r$.
Later in this section we will show that we can always find a solution whose length is bounded by a function in~$r$; thus proving fixed-parameter tractability with~$r$.

Our algorithm will build a family~$\WWW^p_v$ of \rppref{} length-$p$ \abwalks{s}{v} for every length~$p$ and each vertex~$v$ using dynamic programming in a Dijkstra fashion --- that is, it will extend the \walks{} along the arcs of the graph.
To ensure that the \rpprefness{} is maintained in this process we only need to remember the sequence~$\sigma = (c_1, \dots, c_r)$ at the end of the color sequence of any walk~$W$.
So we want to compute for each~$v \in V(G)$ and~$p \in \oneto{\ell}$ the family
\begin{equation}
	\label{eq:walk-table}
	\WWW^{p}_{v} \coloneqq \left\{ \sigma \coloneqq (a_1, \dots, a_{p'})\ \middle\lvert\ 
	\begin{aligned}
		& p' = \min \{p+1, r\} \text{ and $G$ contains an \rppref{}}\\
		& \text{length-$p$ \abwalk{s}{v} whose color sequence ends on } \sigma
	\end{aligned}
	\right\}.
\end{equation}

Note that trivial dynamic programming on these families would blow up the size of such families to~$\bigO(\ncols^r)$, which is too large for our purposes.
Yet,
$\sigma$ restricts the choice in colors for the next~$r$ vertices on the path: 
The path may only continue with a sequence~$\rho$ of colors to which~$\sigma$ is $r$-compatible.
If however, for some sequence~$\rho$ there are multiple sequences in~$\WWW^p_v$ that are~$r$-compatible to~$\rho$, then it suffices to remember only one of them.
We call the remaining family an \emph{ordered representative} for~$\WWW^p_v$ and define it formally as follows.

\begin{definition}[Ordered representative]
	\label{def:ordered-rep}
	Let~$p, r \in \NN$ with~$p \le r$ and let $\WWW$ be a family of sequences of length at most~$p$.
	A subfamily~$\hatty \WWW$ of~$\WWW$ is an \orepresentative{} for~$\WWW$ (written~$\widehat\WWW \orepr^r \WWW$) if the following holds for every sequence~$\rho$ of length at most~$r$:
	If there exists a~$\sigma \in \WWW$ that is $r$-compatible to~$\rho$, then there exists a~$\hatty\sigma \in \hatty\WWW$ that is $r$-compatible to~$\rho$.
\end{definition}

To compute an \orepresentative{} for~$\WWW^p_v$ we make use of an algorithm by \citet{fomin2016representative} to compute representatives of (unordered) set families.
Let us first define (unordered) representatives for families of sets.

\begin{definition}[Unordered representative]
	\label{def:unordered-rep}
	Let~$\FFF$ be a family of~$p$-element sets and~$q \in \NN$.
	A subfamily~$\widehat \FFF$ of~$\FFF$ is a \emph{$q$-representative} for~$\FFF$ (written~$\widehat \FFF \repr^q \FFF$) if the following holds for every set~$Y$ of size at most~$q$:
	If~$\FFF$ contains a set~$X$ disjoint from~$Y$, then~$\widehat \FFF$ contains a set~$\widehat X$ disjoint from~$Y$.
\end{definition}

While \citet{fomin2016representative} state their results for families of independent sets of a matroid, for our purposes, the simpler definition for set families (a special case) suffices.

\begin{proposition}[\cite{fomin2016representative}]
	\label{thm:unordered-rep}
	There is an algorithm that, given a family~$\FFF$ of~$p$-sets over a universe~$U$
	and an integer~$q \in \Nzero$, computes in time
	\ifshort{}$\fi{}
		\textstyle \bigO\big( \abs{\FFF} \cdot \binom{p+q}{p} p^\omega + \abs{\FFF} \cdot \binom{p+q}{q}^{\omega-1} \big)
	\ifshort{}$\fi{}
	a $q$-representative~$\hatty\FFF$ for~$\FFF$ of size at most~$\binom{p+q}{p}$.
\end{proposition}

We will first show how to translate a sequence~$\sigma$ into a corresponding (unordered) set so that we can make use of the concept of representatives for unordered set families.
After that, we are ready to devise an algorithm for \cref{thm:dft-fpt-r}.

Consider two \rppref{} walks~$W_\sigma$ and~$W_\rho$ with color sequences~$\sigma = (a_1, \dots, a_p)$ and~$\rho = (b_1, \dots, b_q)$.
We wish to define two functions~$\pi$ and~$\pi'$ that map color sequences to subsets of~$\colors \times \oneto{r}$ such that
$\sigma$ is $r$-compatible to~$\rho$ if and only if~$\pi(\sigma) \cap \pi'(\rho) = \emptyset$.
By definition, 
$\sigma$ is $r$-compatible to~$\rho$ if and only if
$b_i$ does not equal any of the last~$r-i+1$ entries of~$\sigma$.
Define
\begin{equation}
	\label{eq:pi}
	\begin{aligned}
		\pi'(\rho) &\coloneqq \{ (b_i, i) \mid i \in \oneto{\min\{r, q\}} \} \quad \text{and}\\
		\pi(\sigma) &\coloneqq  \{ (a_j, i) \mid i \in \oneto{r},\, j \in \xtoy{p-(r-i)}{p} \cap \NN \}.
	\end{aligned}
\end{equation}
Then,
$(b_i, i) \notin \pi(\sigma)$ if and only if~$b_i$ does not appear among the last~$r-i+1$ entries of~$\sigma$.
In other words, we have the following.

\begin{observation}
	\label{thm:pi-disjoint}
	A sequence~$\sigma$ is~$r$-compatible to a sequence~$\rho$ if and only if~$\pi(\sigma) \cap \pi'(\rho) = \emptyset$.
\end{observation}
We now have the promised connection between ordered and unordered representatives.

\begin{lemma}
	\label{thm:pi}
	Let~$\WWW$ be a family of $p$-sequences and let~$\FFF \coloneqq \{\pi(\sigma) \mid \sigma \in \WWW\}$.
	If~$\hatty \FFF$ is an~$r$-representative of~$\FFF$, then~$\hatty \WWW \coloneqq \{ \sigma \mid \pi(\sigma) \in \hatty \FFF \}$ is an~\orepresentative{} of~$\WWW$.
\end{lemma}
\begin{proof}
	Consider a sequence~$\rho = (b_1, \dots, b_r)$.
	Suppose that~$\sigma = (a_1, \dots, a_p) \in \WWW$ is $r$-compatible to~$\rho$.
	Then by \cref{thm:pi-disjoint}, $\pi(\sigma)$ is disjoint from $\pi'(\rho)$.
	Therefore, there is a set~$\pi(\hatty\sigma) \in \hatty\FFF$ which is disjoint from~$\pi'(\rho)$;
	Thus by \cref{thm:pi-disjoint}, $\hatty\sigma$ is~$r$-compatible to~$\rho$.
\end{proof}

Consequently, we can use \cref{thm:unordered-rep} to compute \orepresentative{s}.

\newcommand{\thmorderedrep}{%
	\label{thm:ordered-rep}
	There is an algorithm that, given a family~$\WWW$ of $p$-sequences over a universe~$U$
	and an integer~$r \in \Nzero$,
	computes in time
	\ifshort{}%
		$\bigO \big( \abs{\WWW} \cdot (r\cdot e)^r r^\omega + \abs{\WWW} \cdot (r\cdot e)^{(\omega-1)r} \big)$
	\fi{}%
	an \orepresentative{} $\hatty\WWW$ of~$\WWW$ of size at most~$(r\cdot e)^r$.
}
\ifshort
\begin{corollary}[\appref{thm:ordered-rep}]
	\thmorderedrep
\end{corollary}
\fi

\appendixproof{thm:ordered-rep}
{
\begin{proof}
	Note that each~$\sigma \in \WWW$
	is of size~$p' \coloneqq \abs{\pi(\sigma)} = \sum_{i=1}^r \abs{\xtoy{p-r+i}{p} \cap \NN}
	= \max\{r-p,0\}\cdot p + \sum_{i=1}^{\min\{p,r\}} i$.
	Hence, we can use \cref{thm:unordered-rep} to compute an~$r$-representative~$\hatty\FFF$ of~$\FFF \coloneqq \{ \pi(\sigma) \mid \sigma \in \WWW \}$, which corresponds to an \orepresentative{} $\hatty\WWW$ by \cref{thm:pi}.
	Observe that~$p' \le \sum_{i=1}^r i = {r(r+1)/2} \eqqcolon r'$.
	If~$r \ge 3$, then~$r'+r \le r^2$, and we have
	\[
		\abs{\hatty\WWW} = \abs{\hatty\FFF} \le \textstyle\binom{p'+r}{p'} \le \binom{r'+r}{r'} = \binom{r'+r}{r} \le \binom{r^2}{r} < (r^2 \cdot e / r)^r = (r\cdot e)^r,
	\]
	using the folklore inequality~$\binom{z}{r} < (ne/r)^r$ for~$1 \le r \le z$ with~$z\in\NN$.
	If~$r = 2$, then~$\binom{r'+r}{r} = \binom{2+1+2}{2} = 10 < (2e)^2$, and if~$r = 1$, then~$\binom{r'+r}{r} = 2 < e$; thus the size bound holds for all values of~$r$.
	The running time bound follows from plugging the size bound into the running time stated in \cref{thm:unordered-rep} and the fact that~$\abs{\FFF} = \abs{\WWW}$.
\end{proof}
}

Ordered representatives are transitive, just like their unordered counterparts \cite{fomin2016representative}.

\newcommand{\thmtransitive}{%
	\label{thm:transitive}
	If~$\hatty\WWW \orepr^r \tildy\WWW$ and~$\tildy\WWW \orepr^r \WWW$, 
	then~$\hatty\WWW \orepr^r \WWW$.
}
\ifshort
\begin{observation}[\appref{thm:transitive}]
	\thmtransitive
\end{observation}
\fi

\appendixproof{thm:transitive}{
\begin{proof}
	Consider a $q$-sequence~$\rho$ such that some~$\sigma \in \WWW$ is $r$-compatible to~$\rho$.
	By the definition of \orepresentative{s} and since~$\tildy\WWW \orepr^r \WWW$, 
	there exists a~$\tildy\sigma \in \tildy\WWW$ that is $r$-compatible to~$\rho$.
	Then there also exists a~$\hatty\sigma \in \hatty\WWW$ that is $r$-compatible to~$\rho$
	since~$\hatty\WWW \orepr^r \tildy\WWW$.
\end{proof}
}

With a way to efficiently compute \orepresentative{s} at hand, 
we can compute an \rppref{} \stwalk{} of length~$\ell$ with the following routine.
Recall that we are given a graph~$G$ with two terminals~$s$ and~$t$, 
a coloring~$\col \colon V(G) \to C$, 
and two integers~$r$ and~$\ell$ as input.

\begin{alg}\label{alg:fpt-r}
	Set~$\hatty\WWW^0_{s} \ceq \{ (\col(s)) \}$
	and for all~$v \in V(G) \setminus \{s\}$, 
	set~$\hatty\WWW^0_{v} \ceq \emptyset$.
	Now,
	for each~$p = 1, 2, \dots, \ell$ 
	compute for all~$v \in V(G)$ the set
	\begin{equation}
		\label{eq:recurrence}
		\NNN^p_{v} \ceq
		\begin{cases}
			\displaystyle\bigcup_{u \in N^-(v)}
			\left\{(a_1, \dots, a_{p+1}) \ \Big\vert\
				\begin{aligned}
					& (a_1, \dots, a_p) \in \hatty\WWW^{p-1}_u \text{ and }\\
					& \col(v) = a_{p+1} \notin \{a_1, \dots, a_p\}
				\end{aligned}
			\right\}
			& \text{ if } p < r,\\
			\displaystyle\bigcup_{u \in N^-(v)}
			\left\{(a_2, \dots, a_{r+1}) \ \Big\vert\
				\begin{aligned}
					& (a_1, \dots, a_r) \in \hatty\WWW^{p-1}_u \text{ and }\\
					& \col(v) = a_{r+1} \notin \{a_1, \dots, a_r\}
				\end{aligned}
			 \right\}
			& \text{ if } p \ge r,
		\end{cases}
	\end{equation}
	and an \orepresentative{}~$\hatty\WWW^p_{v} \orepr^r \NNN^p_{v}$.
	Return \yes{} if and only if~$\hatty\WWW^{q}_t \ne \emptyset$ for some $q \in \oneto{\ell}$.
\end{alg}

Let us show that \cref{alg:fpt-r} indeed computes representatives of the family~$\WWW^p_{v}$ as defined in equation~\eqref{eq:walk-table}, and hence, is correct.

\begin{lemma}
	\label{lem:fpt-r-repr}
	For each~$v \in V(G)$ and~$p \in \nullto{\ell}$,
	the family $\hatty \WWW^p_{v}$ computed by \cref{alg:fpt-r}
	contains at most~$(r\cdot e)^r$ sets
	and is an \orepresentative{} of $\WWW^p_{v}$ as defined in~\eqref{eq:walk-table}.
\end{lemma}
\begin{proof}
	Our proof is by induction.
	By the initial assignments of $\hatty \WWW^0_{v}$, the statement holds for~$p=0$.
	Now, fix some~$p \in \oneto{\ell}$ and assume that~$\hatty\WWW^{p-1}_u \orepr^r \WWW^{p-1}_u$ for all~$u \in V(G)$.

	Let~$\rho = (b_1, \dots, b_q)$ be a color sequence with~$q \le r$ and let~$v \in V(G)$.
	Suppose that there exists a sequence~$\sigma \in \WWW^p_v$ that is $r$-compatible to~$\rho$.
	We claim that
	there exists a~$\hatty\sigma \in \hatty\WWW^p_v$ that is $r$-compatible to~$\rho$; 
	thus proving that~$\hatty\WWW^p_v \orepr^r \WWW^p_v$.
	As the bound on~$\abs{\hatty\WWW^p_v}$ then follows from \cref{thm:ordered-rep}, we are done once the claim is proven.
	We will prove the claim first for~$p \ge r$ and afterwards for~$p < r$.

	If~$p \ge r$, then~$\sigma$ is an $r$-sequence~$(a_2, \dots, a_{r+1})$
	and there exists an \rppref{} length-$p$ \abwalk{s}{v} $W$ whose color sequence ends on~$\sigma$.
	Let~$a_1$ be the color that~$W$ visits just before visiting the colors in~$\sigma$, that is, the color sequence of~$W$ ends on~$(a_1, \dots, a_{r+1})$.
	Further, let~$u$ be the penultimate vertex visited by~$W$
	and let~$W'$ be the length-$(p-1)$ \subw{} of~$W$ ending on~$u$.
	Then the color sequence of~$W'$ ends on~$\sigma' \coloneqq (a_1, \dots a_r)$.
	Let~$\rho' \coloneqq (a_{r+1}) \circ \rho$.
	Observe that~$\sigma'$ is $r$-compatible to~$\rho'$, due to $\sigma$~being $r$-compatible to~$\rho$ and $W'$ being \rppref{}.
	Thus, by our induction hypothesis and the definition of \orepresentative{s}, there exists a sequence~$\hatty\sigma' \in \hatty\WWW^{p-1}_u$ that is $r$-compatible to~$\rho'$.
	Let~$\hatty W'$ be the \abwalk{s}{u} corresponding to~$\hatty\sigma'$ and let~$\hatty\sigma' \coloneqq (\hatty a'_1, \dots, \hatty a'_r)$.
	Define~$\hatty\sigma \coloneqq (\hatty a'_2, \dots, \hatty a'_r) \circ (a_{r+1})$.
	As~$u \in N^-(v)$ and~$\col(v) = a_{r+1} \notin \{\hatty a'_1, \dots, \hatty a'_r\}$, 
	we have that~$\hatty\sigma \in \NNN^p_v$.
	Finally, observe that~$\hatty\sigma$ is $r$-compatible with~$\rho$.
	Thus, $\NNN^p_v \orepr^r \WWW^p_v$.
	Since $\hatty\WWW^p_v \orepr^r \NNN^p_v$, the claim follows for $p \ge r$ due to the transitivity of \orepresentative{s} (\cref{thm:transitive}).

	If~$p < r$, then~$\sigma$ is a $(p+1)$-sequence~$(a_1, \dots, a_{p+1})$
	and there exists an \rppref{} length-$p$ \abwalk{s}{v} $W$ whose color sequence ends on~$\sigma$.
	Indeed, $\sigma$ is the entire color sequence of~$W$.
	This case is similar to the above, but there is no color that is visited before~$\sigma$ in~$W$.
	Hence, in this case, $\sigma' \coloneqq (a_1, \dots, a_p)$
	and~$\hatty\sigma \coloneqq \hatty\sigma \circ (a_{p+1})$.
	The remainder of the proof is the same.
\end{proof}

Next, 
we show that the algorithm runs in the claimed running time.%
\ifshort{}
\cref{thm:dft-fpt-r} then follows from \cref{lem:fpt-r-repr,lem:fpt-r-time}.
\fi{}

\newcommand{\lemfptrtime}{%
	\label{lem:fpt-r-time}
	\cref{alg:fpt-r} runs in $\bigO((r\cdot e)^{\omega r} \cdot \ell m)$ time on~$m$-arc digraphs.
}
\ifshort
\begin{lemma}[\appref{lem:fpt-r-time}]
	\lemfptrtime
\end{lemma}
\fi
\appendixproof{lem:fpt-r-time}
{
\begin{proof}
	Let~$v \in V(G)$ and~$p \in \oneto{\ell}$.
	Note that all sequences
	in~$\NNN^p_{v}$ have length~$\min\{p+1, r\}$.
	To compute~$\NNN^p_{v}$, we need to consider the sets~$\hatty\WWW^{p-1}_{u}$ for each in-neighbor~$u$ of~$v$.
	By \cref{thm:ordered-rep}, $\hatty\WWW^{p-1}_{u}$ contains at most~$(r\cdot e)^r$ sets.
	Hence, $\NNN^p_{v}$ contains at most~$\deg^-(v) \cdot (r\cdot e)^r$ sets, and its computation time is upper bounded by the same term.
	Thus the time to compute~$\hatty\WWW^p_{v}$ is
	\begin{align*}
		& \bigO \big( \abs{\NNN^p_v} (r\cdot e)^r r^\omega + \abs{\NNN^p_v} (r\cdot e)^{(\omega-1)r} \big)\\
		\subseteq~& \bigO \big( \deg^-(v) (r\cdot e)^{2r} r^\omega + \deg^-(v) (r\cdot e)^{\omega r} \big)\\
		\subseteq~& \bigO \big( \deg^-(v) (r\cdot e)^{\omega r} \big).
	\end{align*}
	Doing this for each~$v \in V(G)$ and each~$p \in \oneto{\ell}$ 
	and assuming that our graph is connected
	yields the claimed running time 
	by the handshaking lemma.
\end{proof}
}

\ifshort{}\fi

\appendixproof{thm:dft-fpt-r}{
\begin{proof}[Proof of \cref{thm:dft-fpt-r}]
	By \cref{lem:fpt-r-repr}, \cref{alg:fpt-r} correctly computes the \orepresentative~$\hatty \WWW^p_{t}$ for the family~$\WWW^p_{t}$ in~\eqref{eq:path-table} for each~$p \in \oneto{\ell}$.
	Hence, the algorithm returns \yes{} if and only if there exists an \rppref{} \stwalk{} of length~$q \le \ell$.
	The running time follows from \cref{lem:fpt-r-time}.
\end{proof}
}

\paragraph{Bounding the length of the walk.}
Note that the length of a walk may be significantly longer than the running time of the above algorithm for \fw{}.
Hence, \cref{thm:dft-fpt-r} does not imply fixed-parameter tractability for the problem of finding an \rppref{} \stwalk{} of \emph{any} length.
We can however show that we can always find an \rppref{} \stwalk{} in which the number of visits to each vertex is bounded by a function in~$r$.
The idea is as follows.
Consider a vertex~$v$ that is visited multiple times by an \rppref{} walk~$W$.
Relevant for us are the consecutive subsequences of length~$r-1$ of the color sequence~$\tau$ of~$W$ that appear immediately before and after each visit of~$v$.
Consider the~$i$-th visit and let~$\sigma_i$ and~$\rho_i$ be the consecutive length-$(r-1)$ subsequences of~$\tau$ before and after the~$i$-th visit of~$v$,
that is, the sequence~$\sigma_i \circ \col(v) \circ \rho_i$ is a consecutive subsequence of~$\tau$.
Now, if for a later visit, say the~$j$-th visit of~$v$, we have that~$\sigma_i \circ \col(v)$ is~$r$-compatible to~$\rho_j$, then we can \emph{skip} all vertices between~$i$ and~$j$.
We will show that the number of visits to~$v$ is bounded by a function in~$r$, or else we can skip visits.
\ifshort{}
For this, we will make use of a skewed variant \cite{frankl1982extremal} of the seminal Bollob{\'a}s' Two Families Theorem~\cite{bollobas1965}.
Details are deferred to \cref{app:sec:short-walk}.
\fi{}

\toappendix{
\subsection{Bounding the Length of the Walk}
\label{app:sec:short-walk}
\begin{proposition}[\citet{frankl1982extremal}]
	\label{thm:bollobas}
	Let~$A_1, \dots, A_m$ and~$B_1, \dots, B_m$ be finite sets such that
	for every~$i,j\in\oneto{m}$,
	$A_i \cap B_i = \emptyset$  
	and~$A_i \cap B_j \ne \emptyset$ if~$i < j$.
	If~$\abs{A_i} \le p$ and~$\abs{B_i} \le q$ for every~$i\in\oneto{m}$,
	then~$m \le \binom{p+q}{p}$.
\end{proposition}

Using \cref{thm:pi-disjoint}, we can prove an analog for sequences.

\begin{corollary}
	\label{thm:ordered-bollobas}
	Let~$\sigma_1, \dots, \sigma_m$ and~$\rho_1, \dots, \rho_m$ be sequences such that 
	for every~$i,j\in\oneto{m}$,
	$\sigma_i$ is $r$-compatible to~$\rho_i$
	but~$\sigma_i$ is \emph{not}~$r$-compatible to~$\rho_j$ if~$i < j$.
	If~$\abs{\sigma_i} \le r$ and~$\abs{\rho_i} \le r$ for every~$i\in\oneto{m}$,
	then~$m \le (r\cdot e)^r$.
\end{corollary}
\begin{proof}
	Let~$A_i \coloneqq \pi(\sigma_i)$ and~$B_i \coloneqq \pi'(\sigma_i)$ (see~\eqref{eq:pi}).
	Then~$\abs{A_i} \le r(r+1)/2 \eqqcolon r'$ and~$\abs{B_i} \le r$.
	Due to~\cref{thm:pi-disjoint},
	we have for every~$i,j\in\oneto{m}$,
	$A_i \cap B_i = \emptyset$  
	and~$A_i \cap B_j \ne \emptyset$ if~$i < j$.
	Thus,
	by \cref{thm:bollobas},
	$m \le \binom{r+r'}{r'}$ which is less than~$(r\cdot e)^r$ as shown in the proof of \cref{thm:ordered-rep}.
\end{proof}

Now we can prove our bound of vertex visits in an \rppref{} walk.

\begin{lemma}
	\label{thm:short-walk}
	Let~$G$ be a digraph and~$s$ and~$t$ be two vertices.
	If~$G$ contains an \rppref{} \stwalk{} $W$, then it contains an \rppref{} \stwalk{} $W^*$ that visits each vertex at most~$((r-1)\cdot e)^{r-1}$ times.
\end{lemma}
\begin{proof}
	Assume towards a contradiction that~$W$ is a shortest \rppref{} \stwalk{} and that there is a vertex~$v$ which is visited~$q > ((r-1) \cdot e)^{r-1}$ times by~$W$.
	For~$i \in \oneto{q}$, consider the subwalks~$W_i$ and~$W'_i$ of~$W$ that contain the~$r-1$ vertices before and after the~$i$-th visit of~$v$, respectively.\footnote{If there are less than~$r-1$ vertices before or after the~$i$-th visit of~$v$, then the corresponding subwalk contains all vertices that appear before or after the~$i$-th visit in~$W$.}
	Denote by~$\sigma_i$ and~$\rho_i$ the color sequences of~$W_i$ and~$W'_i$.
	That is, the color sequence surrounding the~$i$-th visit of~$v$ is~$\sigma_i \circ (\col(v)) \circ \rho_i$.
	As~$W$ is \rppref{}, it is easy to verify that~$\sigma_i$ is~$(r-1)$-compatible to~$\rho_i$.
	As~$v$ is visited more then~$(r\cdot e)^r$ times, \cref{thm:ordered-bollobas} tell us that there exist~$i, j \in \oneto{q}$, $i < j$ such that~$\sigma_i$ is $r$-compatible to~$\rho_j$.
	Let~$W^*_1$ be the subwalk of~$W$ starting in~$s$ and ending before the~$i$-th visit of~$v$ (that is, the color sequence of $W^*_1$ ends on~$\sigma_i$)
	and let~$W^*_2$ be the subwalk of~$W$ starting after the~$j$-th visit of~$v$ (that is, the color sequence of~$W^*_2$ starts with~$\rho_j$).
	Then~$W^* \coloneqq W^*_1 \circ (v) \circ W^*_2$ is \rppref{} and shorter than~$W$ --- a contradiction.
\end{proof}
}

Combining \ifshort{}the above\fi{} with \cref{thm:dft-fpt-r}, we can prove that deciding whether a graph contains an \rppref{} \stwalk{} is fixed-parameter tractable when parameterized by~$r$.

\newcommand{\cordftfptr}{%
	\label{cor:dft-fpt-r}
	Given a vertex-colored $m$-arc digraph and two vertices~$s$ and~$t$, one can decide in
	$\bigO((r\cdot e)^{r(\omega +1)} \cdot m)$ time whether the graph contains an \rppref{} \stwalk{}.
}
\ifshort
\begin{corollary}[\appref{cor:dft-fpt-r}]
	\cordftfptr
\end{corollary}
\fi
\appendixproof{cor:dft-fpt-r}{
\begin{proof}
	Using \cref{thm:dft-fpt-r}, we check whether there exists an \rppref{} \stwalk{} of length at most~$\ell = ((r-1)\cdot e)^{r-1}$ in~$G$.
	The correctness follows from \cref{thm:short-walk}.
	The running time is~$\bigO((r\cdot e)^{\omega r} \cdot ((r-1)e)^{(r-1)} \cdot m) \subseteq \bigO((r\cdot e)^{r(\omega +1)} \cdot m)$.
\end{proof}
}

\subsection{A Matching Lower Bound}

\newcommand{\phs}{$k\times k$~Permutation Hitting Set}
\newcommand{\phsTsc}{\prob{\phs}}
\newcommand{\phsAcr}{\prob{$k\times k$~PHS}}

A close look at the above algorithm shows that using the algorithm by \citet{fomin2016representative} to compute \orepresentative{s} is actually not optimal
as the underlying unordered representative family also stores representatives for any set~$Y$ which does not correspond to~$\pi'(\rho)$ for any sequence~$\rho$.
This raises hope for a more efficient algorithm.
We can however show that finding an algorithm with a running time that is asymptotically smaller in the exponent than ours presented in \cref{ssec:FPTwalks}
would break the Exponential Time Hypothesis (ETH) \cite{ImpagliazzoP01}.
We remark that the reduction also proves the problem to be \NP-hard.
Moreover, 
the provided reduction holds even for shortest walks, 
and thus also for shortest paths.
Altogether, we will prove the following in this section.

\newcommand{\thmreth}{%
	\label{thm:r-eth}
	Even if~$\ell = \dist(s, t)$ and on acyclic digraphs,
	both \fw{}
	and \fp{}
	are \NP-hard
	and,
	unless the ETH fails,
	cannot be solved in~$2^{o(r \log r)} \cdot n^{\bigO(1)}$ time
	on~$n$-vertex digraphs.
}
\ifshort
\begin{theorem}[\appref{thm:r-eth}]
	\thmreth
\end{theorem}
\fi

We will provide a polynomial-time reduction in which the parameter will be bounded linearly from the \phsTsc{} problem,
where one is given a family~$\FFF$ of subsets of a universe~$\oneto{k}\times\oneto{k}$ (which we will treat like a grid with~$k$ rows and columns),
and one is asked whether there is a \emph{hitting permutation}, 
that is, 
a bijection~$\varphi \colon \oneto{k} \to \oneto{k}$ such that 
each~$F \in \FFF$ contains an element~$(i, \varphi(i))$ with~$i \in \oneto{k}$.
Unless the ETH fails, 
\phsTsc{} cannot be solved in~$2^{o(k \log k)} \cdot (k+\abs{\FFF})^{\bigO(1)}$ time \cite{lokshtanov2018slightly}.

The rough idea for the construction is as follows 
(see \cref{fig:r-eth} for an illustrative example):
\begin{figure}[t]
	\centering
	\begin{tikzpicture}
		\def\xr{0.975}
		\def\xxr{0.90}
		\def\yr{.8}
		\def\elemr{.7}

		\tikzstyle{xn}=[scale=0.5,fill=white]
		\tikzstyle{xe}=[arc,semithick] %
		\tikzstyle{xei}=[arc,color=black!40] %
		\tikzstyle{xnA}=[colornode=green,   scale=0.65,thick]
		\tikzstyle{xnB}=[colornode=violet,  scale=0.65,thick]
		\tikzstyle{xnC}=[colornode=blue,    scale=0.65,thick]
		\tikzstyle{xn0}=[colornode=yellow,scale=0.65,thick]

		\tikzstyle{ynA}=[scale=1.75*\elemr,fill=none,rectangle,draw,thick]
		\tikzstyle{ynB}=[scale=1.45*\elemr,fill=none,circle,draw,thick]
		\tikzstyle{ynC}=[scale=1.25*\elemr,fill=none,diamond,draw,thick]

		\newcommand{\LeGridNodes}[3]{
			\foreach \x in {1,2,3}{
				\foreach \y/\lab in {1/xnA,2/xnB,3/xnC}{
					\node (n#1x\x\y) at (#2*\xr+\x*\xr,#3*\yr+\y*\yr)[\lab]{};
				}
			}
		}
		\newcommand{\LeGridAllEdges}[1]{
			\foreach \y in {1,2,3}{
				\foreach \x in {1,2}{
					\foreach \yx in {1,2,3}{
						\pgfmathsetmacro\xt{int(\x+1)}
						\draw[xei] (n#1x\x\y) to (n#1x\xt\yx);
					}
				}
			}
		}
		\newcommand{\LeGrid}[3]{
			\LeGridNodes{#1}{#2}{#3}
			\LeGridAllEdges{#1}
		}

		\newcommand{\hiliLeGrid}[6]{
			\ifnum\pdfstrcmp{#4}{0}=1
				\foreach \yt in {#4}{
					\node at (n#2x1\yt)[#3]{};
				}
			\fi
			\ifnum\pdfstrcmp{#5}{0}=1 
				\foreach \yt in {#5}{
					\node at (n#2x2\yt)[#3]{};
					\foreach \y in {1,2,3}{
						\draw[xe] (n#1x1\y) to (n#2x2\yt);
					}
				}
			\fi
			\ifnum\pdfstrcmp{#6}{0}=1
				\foreach \yt in {#6}{
					\node at (n#2x3\yt)[#3]{};
					\foreach \y in {1,2,3}{
						\draw[xe] (n#1x2\y) to (n#2x3\yt);
					}
				}
			\fi
		}

		\newcommand{\LeGridToNode}[2]{
			\foreach \y in {1,2,3}{
				\draw[xei] (n#1x3\y) to (#2);
			}
		}
		\newcommand{\NodeToLeGrid}[4]{
			\foreach \y in {1,2,3}{
				\draw[xei] (#1) to (n#2x1\y);
			}
			\ifnum\pdfstrcmp{#4}{0}=1
				\foreach \y in {#4}{
					\draw[xe] (#1) to (n#3x1\y);
				}
			\fi
		}

		\def\ylab{5.75}
		\node (A1) at (0,\ylab*\yr)[ynA]{};
		\node at (A1.east)[xshift=-.3ex,anchor=west,font=\small]{$=\{(1,2),(2,2)\}$};

		\node (A2) at (3.9*\xr,\ylab*\yr)[ynB]{};
		\node at (A2.east)[xshift=-.3ex,anchor=west,font=\small]{$=\{(1,1),(2,2),(2,3),(3,3)\}$};

		\node (A3) at (10.2*\xr,\ylab*\yr)[ynC]{};
		\node at (A3.east)[xshift=-.3ex,anchor=west,font=\small]{$=\{(2,1),(3,1),(3,2)\}$};

		\node (u1) at (0*\xr,2*\yr)[xn0]{};
		\LeGrid{Au}{0}{2}
		\LeGrid{Al}{0}{-2}
		\hiliLeGrid{Al}{Au}{ynA}{2}{2}{0}
		\NodeToLeGrid{u1}{Al}{Au}{2}
		\node (u2) at (4*\xr,2*\yr)[xn0]{};
		\LeGridToNode{Au}{u2}
		\LeGrid{Bu}{4}{2}
		\LeGrid{Bl}{4}{-2}
		\hiliLeGrid{Bl}{Bu}{ynB}{1}{2,3}{3}
		\NodeToLeGrid{u2}{Bl}{Bu}{1}
		\node (u3) at (8*\xr,2*\yr)[xn0]{};
		\LeGridToNode{Bu}{u3}
		\node (unix) at (9*\xr,2*\yr)[]{$\cdots$};
		\node (um) at (10*\xr,2*\yr)[xn0]{};
		\LeGrid{Cu}{10}{2}
		\LeGrid{Cl}{10}{-2}
		\hiliLeGrid{Cl}{Cu}{ynC}{0}{1}{1,2}
		\NodeToLeGrid{um}{Cl}{Cu}{0}
		\node (umpo) at (14*\xr,2*\yr)[xn0]{};
		\LeGridToNode{Cu}{umpo}

		\node at (u1)   [label=-90:{$s$}] {};
		\node at (u2)   [label=-90:{$u_2$}] {};
		\node at (u3)   [label=-90:{$u_3$}] {};
		\node at (um)   [label=-90:{$u_m$}] {};
		\node at (umpo) [label=-90:{$t$}] {};

		\begin{scope}[on background layer]
			\tikzstyle{lpath}=[line width=9pt,color=black!10,line cap=round,line join=round]
			\draw[lpath]
				(u1.center) 
				to (nAux12.center) to (nAux21.center) to (nAux33.center)
				to (u2.center)
				to (nBlx12.center) to (nBlx21.center) to (nBux33.center)
				to (u3.center);
			\draw[lpath]
				(um.center) 
				to (nClx12.center) to (nCux21.center) to (nCux33.center)
				to (umpo.center);
		\end{scope}
		
		\node at (nAux21)[label=-90:$w^{2,1}_1$]{};
		\node at (nBux33)[label=0:$w^{3,3}_2$]{}; 
		\node at (nBlx12)[label=180:$v^{1,2}_2$]{}; 
	\end{tikzpicture}
	\caption[Illustrative example for~\cref{constr:r-eth}.]{
		Example for~\cref{constr:r-eth} on a $3\times 3$-grid with
		$\FFF = \{\tikz{\node[scale=0.9,rectangle,draw,semithick] {};}, \tikz{\node[scale=0.66,circle,draw,semithick] {};}, \dots, \tikz{\node[scale=0.5,diamond,draw,semithick] {};}\}$.
		Gray (thin) arcs exist independently of the current subset~$F_i \in \FFF$.
		Black arcs point to elements in~$F_i$ and are the only way to reach the top copy.
		The highlighted path selects the hitting set~$\{(1,2), (2,1), (3,3)\}$,
		thereby visiting, i.a., 
		$w^{2,1}_1$, $v^{1,2}_2$, and $w^{3,3}_2$.
		The path visits one black arc for each~$F_i \in \FFF$.
	}
	\label{fig:r-eth}
\end{figure}
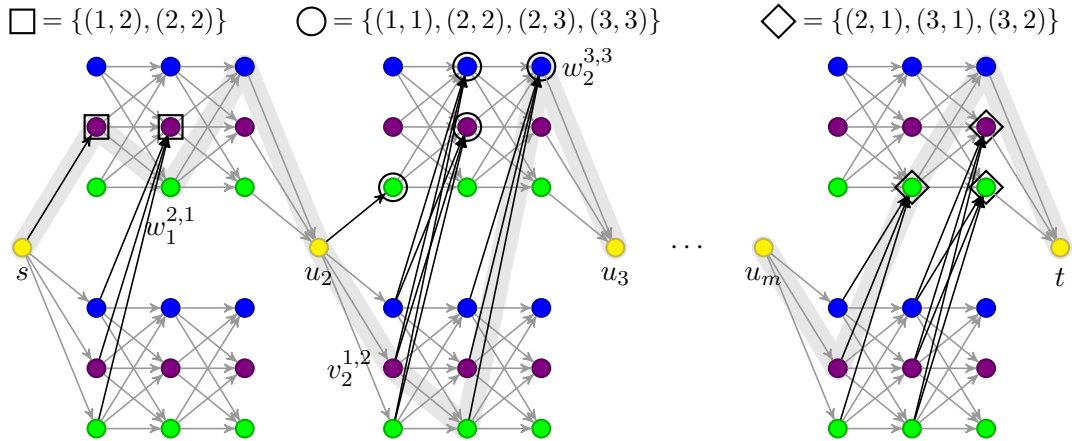
For each set in~$\FFF$,
we create a pair of copies 
(lower and upper) 
of our universe, which we will be able to traverse column by column.
Each row receives a color, 
and the subpath length is chosen such that we always have to visit the colors in the same order for each set in~$\FFF$.
Hence, 
for each~$F \in \FFF$, 
we must pick the same permutation.
Now, 
we always start ``left'' of the two copies for~$F$,
and we can always go to the lower copy,
but in order to get to the next set in~$\FFF$, 
we need to get to the upper copy;
This is only possible if one element from our permutation is in~$F$.
Hence, 
there is an \rppref{} \stwalk{} 
(indeed, by construction, it will always be a path) 
if and only if there is a hitting permutation for~$\FFF$.

We now describe the construction formally for \fw{}.
Note that the construction for \fp{} is identical.

\begin{construction}
	\label{constr:r-eth}
	Given an instance~$I = (\oneto{k}\times\oneto{k}, \FFF)$ of \phsTsc{}
	with~$\FFF = \{F_1, \dots, F_m\}$,
	we construct an instance~$I' = (G, \col, s, t, r, \ell)$ of \fw{} 
	with~$C \coloneqq \oneto{k+1}$, 
	$r \coloneqq k$, 
	and
	$\ell \coloneqq m (k+1)$ as follows
	(see \cref{fig:r-eth}).
	We add for each~$q \in \oneto{m+1}$ a vertex~$u_q$
	and set~$s \coloneqq u_1$ and~$t \coloneqq u_{k+1}$.
	Then, for each~$q \in \oneto{m}$, we do the following.
	For each~$(i, j) \in \oneto{k}\times\oneto{k}$, we add the vertices~$v^{i,j}_q$ and~$w^{i,j}_q$ with color~$j$ to~$G$.
	For each~$j \in \oneto{k}$, 
	we add the arcs~$(u_q,v^{1,j}_q)$ and~$(w^{k,j}_q, u_{q+1})$ to~$G$.
	Moreover, 
	for each~$j \in \oneto{k}$ with~$(1,j) \in F_q$, 
	we add the arc~$(u_q, w^{1,j}_q)$ to~$G$.
	For each~$i \in \xtoy{2}{k}$ and each~$j \in \oneto{k}$ with~$(i,j) \in F_q$, 
	we add the arcs~$(v^{i-1,j'}_q, w^{i,j}_q)$ for each~$j' \in \oneto{k}$.
	Finally, 
	for each~$q\in\oneto{m}$,
	for each~$j\in\oneto{k}$,
	and for each~$i\in\oneto{k-1}$,
	add the arcs~$(v_q^{i,j},v_q^{i+1,j})$ and $(w_q^{i,j},w_q^{i+1,j})$.
\end{construction}

We now prove the deciding property of any \rppref{} \stwalk{} in our constructed instance.

\newcommand{\thmpermutationwalk}{%
	\label{thm:permutation-walk}
	If graph~$G$ from an instance $(G,c,s,t,r,\ell)$ obtained via~\cref{constr:r-eth} contains an \rppref{} \stwalk{}~$W$ of length~$\ell$, 
	then~$W$ is a path
	and there is a bijection~$\varphi \colon \oneto{k} \to \oneto{k}$
	such that for each~$q \in \oneto{m}$ and~$i \in \oneto{k}$, 
	$W$ visits either~$v^{i, \varphi(i)}_q$ or~$w^{i, \varphi(i)}_q$.
}
\ifshort
\begin{observation}[\appref{thm:permutation-walk}]
	\thmpermutationwalk
\end{observation}
\fi

\appendixproof{thm:permutation-walk}
{
\begin{proof}
	Observe that~$\dist(s, t) = \ell$, hence any \stwalk{} is a shortest walk and thus a path.

	As there are~$k+1$ colors and~$r=k$, 
	the walk~$W$ must always traverse the colors in the same order,
	otherwise there exists a non-colorful subpath of length~$r$.
	As the only vertices with color~$k+1$ are~$u_q$, $q \in \oneto{m}$, 
	we have by construction that~$W$ will visit exactly one vertex out of~$\bigcup_{j \in \oneto{k}} \{v^{i,j}_q, w^{i,j}_q\}$ for each~$q \in \oneto{m}$ and~$i \in \oneto{k}$.
	Also, 
	for each~$q \in \oneto{m}$, 
	the path will visit the vertex in column~$i$ before visiting the vertex in column~$i+1$, for each~$i \in \oneto{k-1}$.
	Suppose now that there is no bijection as claimed.
	Then there exists a~$q \in \oneto{m-1}$ and~$i \in \oneto{k}$ such that~$W$ visits one of~$v^{i,j}_q$ and~$w^{i,j}_q$ 
	(call this vertex~$x$) 
	and one of~$v^{i,j'}_{q+1}$ and~$w^{i,j'}_{q+1}$ 
	(call this vertex~$y$), 
	for distinct~$j, j' \in \oneto{k}$.
	As the vertices between~$u_q$ and~$u_{q+1}$ 
	(and between~$u_{q+1}$ and~$u_{q+2}$) 
	must be visited column by column, 
	the subpath from~$x$ to~$y$ is of length~$r$.
	But then, 
	as~$\col(x) = \col(y)$, 
	$W$ is not \rppref{} --- a contradiction.
\end{proof}
}

With this property at hand, proving the theorem of this section is straightforward.
\appendixproof{thm:r-eth}
{
\begin{proof}[Proof of \cref{thm:r-eth}]
	We provide a polynomial-time reduction from \phsTsc{}.
	Let $I' = (G, \col, s, t, r, \ell)$ be the instance obtained via 
	\cref{constr:r-eth} on input instance~$I = (\oneto{k}\times\oneto{k}, \FFF)$.
	We claim that~$I$ is a \yes-instance if and only if~$I'$ is a \yes-instance.

	Suppose first that~$I$ is a \yes-instance and let~$\varphi \colon \oneto{k} \to \oneto{k}$ be a hitting permutation.
	For each~$q \in \oneto{m}$ let~$i_q \coloneqq \min \{i \in \oneto{k} \mid (i, \varphi(i)) \in F_q\}$.
	We claim that
	\[
		W \coloneqq \bigcirc_{q=1}^m \left( (u_q) \circ \left( \bigcirc_{i=1}^{i_q-1} (v^{i,\varphi(i)}_q) \circ \bigcirc_{i=i_q}^k (w^{i, \varphi(i)}_q) \right) \right) \circ (u_{m+1})
	\]
	is an \rppref{} \stwalk{} of length~$\ell$.
	As~$\ell = \dist(s, t)$, $W$~a shortest \stwalk{}, and thus a shortest \stpath{}.
	Clearly,
	$W$ has length~$\ell$.
	Further, $W$ visits~$u_q$ with color~$k+1$ before visiting~$k$ vertices, of which the~$i$-th vertex always has color~$\varphi(i)$; thus it is \rppref{}.
	The crucial part in proving that~$W$ is an \stwalk{} is to show that 
	the arc to~$w^{i_q,\varphi(i_q)}_q$ from its predecessor exists for each~$q \in \oneto{m}$.
	Clearly,
	the remaining arcs of~$W$ do exist.
	So suppose that there is a~$q \in \oneto{m}$ for which there is no arc to~$w^{i_q,\varphi(i_q)}_q$ from its predecessor.
	As the predecessor is~$v^{i_q-1, \varphi(i_q-1)}_q$ 
	if~$i_q > 1$ and~$u_q$ otherwise,
	the arc exists if and only if~$(i_q, \varphi(i_q)) \in F_q$ --- a contradiction to the choice of~$i_q$.
	Thus, 
	$W$ is an \stwalk{} and $I'$ is a \yes-instance.

	Suppose next that~$I'$ is a \yes-instance and let~$W$ be a solution with the properties described in \cref{thm:permutation-walk}.
	Note that this implies that~$W$ is a path.
	We claim that the corresponding bijection~$\varphi$ is a hitting permutation.
	Suppose not, that is, 
	there exists a~$q \in \oneto{m}$ such that~$(i, \varphi(i)) \notin F_q$ for every~$i\in\oneto{k}$.
	Since~$W$ is an \stwalk{},
	by construction $W$ visits~$w^{i, \varphi(i)}_q$ for at least one~$i\in\oneto{k}$.
	Let~$i_q \in \oneto{k}$ be the smallest number such that~$W$ visits~$w^{i_q, \varphi(i_q)}_q$.
	As the predecessor of~$w^{i_q, \varphi(i_q)}_q$ is~$v^{i_q-1, \varphi(i_q-1)}_q$ if~$i_q > 1$ and~$u_q$ otherwise, 
	we have by construction that~$(i_q, \varphi(i_q)) \in F_q$ --- a contradiction.
	Thus, $I$ is a \yes-instance.

	Towards the running time lower bound,
	note that~$n \in \bigO(k^2 \cdot \abs{\FFF})$, 
	$r=k$, 
	and the construction can be computed in~$(k+\abs{\FFF})^{\bigO(1)}$ time.
	Hence, 
	any algorithm for \fw{} running in~$2^{o(r \log r)} \cdot n^{\bigO(1)}$ time
	gives an algorithm for \phsTsc{} running in time~$2^{o(k \log k)} \cdot (k+\abs{\FFF})^{\bigO(1)}$, 
	which breaks the ETH.
\end{proof}
}

Remarkably, 
due to \cref{thm:permutation-walk}, 
the ETH lower bound holds even if one asks whether there exists an \rppref{} \stwalk{} of arbitrary length.
This complements the fixed-parameter tractability of \cref{cor:dft-fpt-r}.

\begin{corollary}
	\label{thm:eth-anylength}
	Unless the ETH breaks, 
	there is no~$2^{o(r \log r)} \cdot n^{\bigO(1)}$-time algorithm for
	the problem of deciding whether there is an \rppref{} \stwalk{} of arbitrary length in a given vertex-colored acyclic digraph with $n$ vertices.
\end{corollary}

Finally, 
as by \cref{thm:permutation-walk} every \rppref{} \stwalk{} in the constructed instance is a shortest \stpath{}, 
we can add to every arc~$(u,v)$ its antiparallel arc~$(v,u)$.
The resulting graph thus is symmetric.

\begin{corollary}
	\label{thm:eth-dags}
	Even if~$\ell = \dist(s, t)$ and on symmetric digraphs,
	both \fw{} and \fp{}
	are \NP-hard
	and, 
	unless the ETH breaks,
	cannot be solved in~$2^{o(r \log r)} \cdot n^{\bigO(1)}$-time on~$n$-vertex digraphs.
\end{corollary}

\section{Paths}
\label{ft:sec:paths}
\appendixsection{ft:sec:paths}

In this section, we study the parameterized complexity of \fp{} with respect to the locality parameter~$r$ and the detour length~$k \coloneqq \ell - \dist(s, t)$.

\subsection{NP-Hardness for Constant Locality Values}
We now provide a dichotomy for \fp{} parameterized by the locality parameter~$r$.
Obviously, if $r = 0$ then any \stpath{} is a solution.
We will now show that the problem remains efficiently solvable when~$r \le 2$, but prove \NP-hardness for all values~$r \ge 3$.

Clearly, if~$r > 0$, we can assume that there is no arc~$(u, v)$ with~$\col(u) = \col(v)$ in our digraph.
Thus, the task of finding a $1$-\pref{} \stpath{} (or \stwalk{}) reduces to finding any \stpath{}.
\begin{observation}
	\label{thm:r=1}
	Finding a shortest $1$-\pref{} \stwalk{} or \stpath{} is linear-time solvable.
\end{observation}

As soon as~$r \ge 2$, the problem becomes much harder.

\newcommand{\thmufp}{%
	\label{thm:ufp}
	\fp{} is \NP-hard for any fixed value of~$r \ge 2$.
}
\ifshort
\begin{theorem}[\appref{thm:ufp}]
	\thmufp
\end{theorem}
\fi

We provide a polynomial-time reduction from \prob{3-Sat},
where given a Boolean formula~$\phi$ in conjunctive normal form such that each clause contains exactly three literals (3-CNF),
the question is whether there exists a truth assignment to the variables for which~$\phi$ evaluates to true.
The problem is known to be \NP-hard, even if each variable appears exactly twice positive and twice negative in the given formula \cite[Theorem 1]{berman2003hardness}.
We provide our construction for~$r=2$ and describe afterwards how it can be adapted to the case when~$r > 2$.

In a nutshell, our construction works as follows.
Our path first needs to go through the variable gadgets, in which there are two branches (for true and false) for each variable.
Afterwards, it needs to go through the clause gadgets, in which there is a branch for each literal.
Each branch visits a vertex of the corresponding variable gadget.
As we are looking for a path, this vertex must be on the branch that was not yet visited by our path.
Finally, the colors in the graph are chosen such that taking any forbidden turn (e.g., from a variable gadget directly into a clause gadget) would breach the \pprefness{} constraint.

\begin{construction}
	\label{constr:ufp}
	Let~$\phi$ be a Boolean formula in 3-CNF in which every variable appears exactly twice positive and twice negative.
	Let~$x_i$, $i \in \oneto{n}$, be the variables
	and let~$c_j$, $j \in \oneto{m}$ be the clauses of~$\phi$.
	We construct an instance~$I' = (G, \col, s, t, r, \ell)$ of \fp{} with~$r=2$, $\ell = 6n + 3m + 2$ and color set~$C \coloneqq \{1, \dots, 4\}$ as follows.

	For each $i \in \oneto{n}$, 
	we build a \emph{variable gadget} as shown in \cref{fig:3sat:1}(a).
	\begin{figure}[t]
		\centering
		\begin{tikzpicture}[xscale=1.125,yscale=1.05]
			\begin{scope}
					\foreach \l/\ll/\lc[count=\i from 0] in {
						$v_i$/$1$/blue,
						$v_i^1$/$2$/green,
						$v_i^2$/$3$/violet,
						{}/$1$/blue,
						$\bar{v_i}^2$/$3$/violet,
						$\bar{v_i}^1$/$2$/green%
					} {
						\node (v\i) at (90+\i*60 : 1)[smallcolornode=\lc, label={90+\i*60}:\l, label={-90+\i*60}:\ll] {};
					}
					\node at (v3)[label={[yshift=-3pt]0:{$y_i$}}]{};
					\foreach \x in {0,1,2}{
						\pgfmathtruncatemacro\y{\x+1}
						\draw[arc] (v\x) to (v\y);
					}
					\foreach \x in {5,4}{
						\pgfmathtruncatemacro\y{\x-1}
						\draw[arc] (v\x) to (v\y);
					}
					\draw[arc] (v0) to (v5);
				\node (v6) at ($(v3)+(0,-0.6)$)[smallcolornode=green, label=left:$2$] {};
				\node (v7) at ($(v6)+(0,-0.6)$)[smallcolornode=violet, label=left:$3$, label=right:$v_i'$] {};
				\draw[arc] (v3) to (v6);
				\draw[arc] (v6) to (v7);
				\node at (-1.3,1.4) {(a)};
			\end{scope}
			\begin{scope}[xshift=3.35cm]
			\def\yr{0.9}
			\def\ysh{1.1}
			\def\xsh{1}
			\node (v0) at (0, 0+\ysh*\yr) [smallcolornode=blue, label=right:$1$, label=above:$w_j$]{};
			\node (v1) at	(0, -3*\yr+\ysh*\yr) [smallcolornode=blue, label=right:$1$, label=below:$w_j'$]{};
			\foreach \y in {0,1}{
				\foreach \x in {0, 1, 2} {
					\ifnum\y=0
						\node (v\x\y) at (\x*\xsh-\xsh,-\y*\yr-1*\yr+\ysh*\yr)[smallcolornode=violet, label=left:$3$] {};
					\fi
					\ifnum\y=1
						\node (v\x\y) at (\x*\xsh-\xsh,-\y*\yr-1*\yr+\ysh*\yr)[smallcolornode=green, label=left:$2$] {};
					\fi
				}
			}
			\foreach \x in {0, 1, 2} {
				\draw[arc] (v0) to (v\x0);
				\draw[arc] (v\x0) to (v\x1);
				\draw[arc] (v\x1) to (v1);
			}
			\node at (-1.3,1.4) {(b)};
			\end{scope}
			\begin{scope}[xscale=0.9,yscale=0.85,xshift=6.35cm,yshift=0.25cm]
			 \begin{scope}
					\begin{scope}
							\foreach \l/\ll/\lc/\dr[count=\i from 0] in {
								/$1$/blue/1,
								$v_i^1$/$2$/green/0,
								$v_i^2$/$3$/violet/0,
								{}/$1$/blue/1,
								$\bar{v_i}^2$/$3$/violet/1,
								$\bar{v_i}^1$/$2$/green/1%
							} {
								\ifnum\dr=1 
									\node (v\i) at (90+\i*60 : 1)[smallcolornode=\lc, label={90+\i*60}:\l, label={-90+\i*60}:\ll] {};
								\fi
							}
							\node at (v0)[label=135:$v_i$]{};
							\foreach \x in {5,4}{
								\pgfmathtruncatemacro\y{\x-1}
								\draw[arc] (v\x) to (v\y);
							}
							\draw[arc] (v0) to (v5);
							\draw[arc,dashed,color=gray] ($(v0)+(0,0.5)$) to (v0);
							\draw[arc,dashed,color=gray] (v0) to ($(v0)+(-0.5,-0.25)$);
							\draw[arc,dashed,color=gray] ($(v3)+(-0.5,0.25)$) to (v3);
							\draw[arc,dashed,color=gray] (v3) to ($(v3)+(0,-0.5)$);
						\node at (-1.,1.4) {(c)};
					\end{scope}
					\begin{scope}[xshift=2.25cm,yshift=-1cm]
					\def\yr{0.9}
					\def\ysh{1.1}
					\def\xsh{1}
					\node (v0) at (0, 0+\ysh*\yr) [smallcolornode=blue, label=right:$1$, label=165:$w_j$]{};
					\node (v1) at	(0, -3*\yr+\ysh*\yr) [smallcolornode=blue, label=right:$1$, label=-165:$w_j'$]{};
					\foreach \y in {0,1}{
						\foreach \x in {0, 1, 2} {
							\ifnum\y=0
								\ifnum\x>0 
									\node (v\x\y) at (\x*\xsh-\xsh,-\y*\yr-1*\yr+\ysh*\yr)[smallcolornode=violet, label=left:$3$] {};
								\fi
							\fi
							\ifnum\y=1
								\node (v\x\y) at (\x*\xsh-\xsh,-\y*\yr-1*\yr+\ysh*\yr)[smallcolornode=green, label=left:$2$] {};
							\fi
						}
					}
					\foreach \x in {0, 1, 2} {
						\ifnum\x>0 
							\draw[arc] (v0) to (v\x0);
							\draw[arc] (v\x0) to (v\x1);
						\fi
						\draw[arc] (v\x1) to (v1);
						\draw[arc,dashed,color=gray] ($(v0)+(\x*0.5-0.5,0.75)$) to (v0);
						\draw[arc,dashed,color=gray] (v1) to ($(v1)+(\x*0.5-0.5,-0.75)$);
					}
					\draw[arc] (v0) to (v4);
					\draw[arc] (v4) to (v01);
					\end{scope}
				\end{scope}
				\begin{scope}[xshift=4.7cm]
					\def\yr{1}
					\node (vn') at (0,0.5)[smallcolornode=violet, label=right:$3$, label=left:$v_n'$] {};
					\node (v0) at ($(vn')+(0,-\yr)$)[smallcolornode=yellow, label=right:$4$] {};
					\node (v1) at ($(v0)+(0,-\yr)$)[smallcolornode=green, label=right:$2$] {};
					\node (w1) at ($(v1)+(0,-\yr)$)[smallcolornode=blue, label=right:$1$, label=left:$w_1$]  {};
					\draw[arc,dashed,color=gray] ($(vn')+(0,0.75)$) to (vn');
					\draw[arc] (vn') to (v0);
					\draw[arc] (v0) to (v1);
					\draw[arc] (v1) to (w1);
					\foreach \x in {0, 1, 2} {
						\draw[arc,dashed,color=gray] (w1) to ($(w1)+(\x*0.5-0.5,-0.75)$);
					}
					\node at (-0.75,1.4){(d)};
				\end{scope}
			\end{scope}

		\end{tikzpicture}
		\caption{(a) The variable gadget and (b) the clause gadget in \cref{constr:ufp}.
		(c) An example showing how a literal path corresponding to literal $\bar x_i$ in clause $c_j$ is attached to the variable gadget at~$\bar v_i^1$.
		(d) The connection between the last variable gadget and the first clause gadget.
		}
		\label{fig:3sat:1}
	\end{figure}
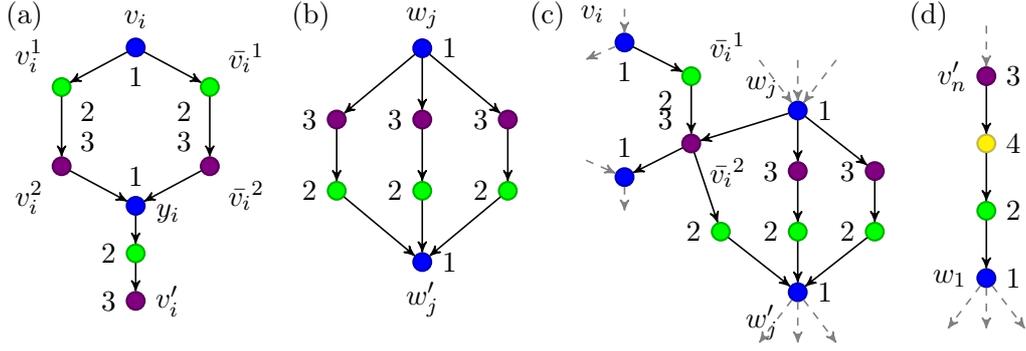
	It contains vertices~$v_i$, $v'_i$, and $y_i$ with colors~$c(v_i)=c(y_i)=1$ and~$c(v_i')=3$.
	Vertices~$v_i$ and~$y_i$ are connected by two parallel length-2 paths~$(v_i,v_i^1,v_i^2,y_i)$ and~$(v_i,\bar v_i^1,\bar v_i^2,y_i)$
	with colors~$c(v_i^1)=c(\bar v_i^1)=2$ and~$c(v_i^2)=c(\bar v_i^2)=3$.
	Vertices~$y_i$ and~$v'_i$ are connected by a length-2 path where the inner vertex has color ``$2$''.
	We add for each~$i \in \oneto{n-1}$ an arc from~$v'_i$ to~$v_{i+1}$.

	Next, 
	for each~$j \in \oneto{m}$,
	we build a \emph{clause gadget} as shown in \cref{fig:3sat:1}(b), 
	which consists of two vertices~$w_j, w'_j$ each with color ``$1$'' 
	connected by three parallel length-2 paths (called \emph{literal paths}), 
	each with internal vertex-color sequence~$(3, 2)$
	and corresponding to one literal.
	For each literal~$\lambda$ of~$c_j$, 
	we pick either~$v_i^1$ or~$v_i^2$ if~$\lambda = x_i$, 
	and either~$\bar v_i^1$ or~$\bar v_i^2$ if~$\lambda = \bar x_i$.
	Since each literal appears at most twice,
	every vertex can be picked at most once.
	We identify the picked vertex with the equally-colored vertices of the literal path of~$\lambda$
	(see \cref{fig:3sat:1}(c) for an example).
	We call the picked vertex the \emph{attached vertex} in the literal path.
	Afterwards, 
	we identify the vertices~$w'_j$ and~$w_{j+1}$ for each~$j \in \oneto{m-1}$.
	We connect~$v'_n$ to~$w_1$ via a length-3 path 
	where the vertex adjacent to~$v'_n$ has color ``$4$''  and the vertex adjacent to~$w_1$ has color ``$2$''
	(see \cref{fig:3sat:1}(d)).
	Finally, we set~$s \coloneqq v_1$ and $t \coloneqq w'_m$.
\end{construction}

\ifshort{}\fi{}

\appendixproof{thm:ufp}{
\begin{proof}[Proof of \cref{thm:ufp}]
	We provide a polynomial-time reduction from the \NP-hard variant of \prob{3-Sat}, in which the Boolean formula~$\phi$ contains each variable exactly twice positive and twice negative \cite{darmann2021simplified}.
	Given such a formula~$\phi$, we use \cref{constr:ufp} to obtain an instance~$I' = (G, \col, s, t, r, \ell)$ of \fp{}.
	Clearly, \cref{constr:ufp} runs in polynomial time.
	It remains to prove that~$\phi$ is satisfiable if and only if~$I'$ is a \yes-instance.

	To this end, assume first that there is some assignment satisfying $\phi$.
	Define~$\alpha \colon \{ x_i, \bar x_i \mid i \in \oneto{n} \} \to \{\true, \false\}$ such that~$\alpha(x_i) = \true$ and~$\alpha(\bar x_i) = \false$ if the satisfying assignment sets~$x_i$ to $\true$
	and~$\alpha(x_i) = \false$ and~$\alpha(\bar x_i) = \true$ otherwise.
	We construct an \stpath{} as follows.
	Starting at~$s = v_1$, for each variable gadget, whenever
	$\alpha(x_i) = \true$, then we go from~$v_i$ to~$v'_i$ via $\bar v_i^1$ and~$\bar v_i^2$.
	Conversely, if 
	$\alpha(x_i) = \false$, then we go via $v_i^1$ and~$v_i^2$.
	Reaching~$v'_n$ in this fashion, we continue to~$w_1$.
	Now, for each clause~$c_j$, we pick some literal~$\lambda$ of~$c_j$ with~$\alpha(\lambda) = \true$ and traverse the corresponding literal path to reach~$w'_j = w_{j+1}$.
	Observe that we do not visit any vertex more than once, due to our choice of path within the variable gadgets.
	It is clear that in this way we can reach~$w'_m$ and thus~$t$, and that the path has length~$6n+3m+2$.
	It is also easy to verify that the resulting path~$P$ is $3$-\pref{}:
	Up until~$v'_n$, the vertices of~$P$ are always colored~$1,2,3,1,2,3,\dots$.
	After traversing the unique vertex of color ``$4$'', the remaining vertices of~$P$ are colored according to the sequence~$2,1,3,2,1,3,\dots$.

	For the converse direction, let~$P$ be any $3$-\pref{} path from~$s$ to~$t$.
	It is not difficult to verify that the path must first traverse all variable gadgets in order and then all clause gadgets in order and has length~$6n+3m+2$.
	Define the variable assignment~$\alpha$ by setting~$\alpha(x_i) = \true$ if and only if~$P$ visits~$\bar v_i^1$ (and then also~$\bar v_i^2$) on its way from~$v_i$ to~$v'_i$.
	We claim that each clause~$c_j$, $j \in \oneto{m}$ is then satisfied by the literal~$\lambda$ of~$c_j$ whose literal path is traversed by~$P$.
	Say without loss of generality~$\lambda = x_i$.
	If we had~$\alpha(x_i) = \false$, then~$P$ would have traversed~$v_i^1$ and~$v_i^2$ on its way from~$v_i$ to~$v'_i$, according to our choice of~$\alpha$.
	But since the literal path of~$\lambda$ again uses either~$v_i^1$ or~$v_i^2$, $P$~cannot be a path --- a contradiction.

	To see that the problem is \NP-hard for any~$r > 2$, one can change the color of all~$y_i$ to~``$4$''
	and change the paths from~$y_i$ to~$v_i'$ to have length~$r-3$, 
	with the vertices being colored $4, 5, \dots, r$.
	Analogously, the literal paths in the clause gadgets are made longer, the vertices being colored according to the sequence~$1,r,r-1,\dots,3,2,1$.
	Finally, the unique vertex that received color~$4$ in the original construction must be recolored to color~$r+1$.
\end{proof}
}
\newcommand{\pecp}{\prob{Properly Edge-Colored Path}}
We close this section by remarking that, on symmetric digraphs, finding a shortest $2$-\pref{} path becomes efficiently solvable.
(The case $r \geq 3$ remains \NP-hard by a reduction similar to the one above.)
The idea here is to transform the vertex coloring into an edge coloring (i.e., every symmetric arc is assigned one color).
We say that a walk is \emph{properly colored} (with respect to some edge coloring) if no two consecutive symmetric arcs share the same color.

\newcommand{\thmedgecoloringlemma}{%
	\label{thm:edge-coloring-lemma}
	Let $G$ be a symmetric digraph, $\col$ a vertex coloring, and $W$ an \stwalk{}.
	Assume that no two adjacent vertices have the same color.
	Then,  $W$~is $2$-\pref{} if and only if it is properly colored with respect to the edge coloring
	$\col'((u, v)) := \{\col(u), \col(v)\}$.
}
\ifshort
\begin{lemma}[\appref{thm:edge-coloring-lemma}]
	\thmedgecoloringlemma
\end{lemma}
\fi
\appendixproof{thm:edge-coloring-lemma}{
\begin{proof}
	The forward direction is immediate.
	Conversely, if $W$ is properly colored with respect to~$\col'$,
	and $(u, v)$, $(v, w)$ are any two consecutive edges of~$W$,
	then we must have $\{\col(u), \col(v)\} \neq \{\col(v), \col(w)\}$
	and therefore $\col(u) \neq \col(w)$.
	Thus, $W$~is then $2$-\pref{}.
\end{proof}
}

As a properly edge-colored \stpath{} can be found in linear time in symmetric digraphs \cite[Cor.~10]{szeider2003transitions}, 
we obtain the following.

\begin{observation}
	\label{thm:r=2}
	Finding a shortest $2$-\pref{} \stwalk{} in a symmetric digraph is solvable in linear time.
\end{observation}

\subsection{Fixed-Parameter Tractability with Detour Length}

We now prove our problem to be fixed-parameter tractable with respect to~$r+k$
where $k$ denotes the length of a \emph{detour} the path may take
(i.e., the desired length~$\ell$ is~$\dist(s, t) + k$).
\ifshort{}\fi{}

\ifshort{}\fi{}

\ifshort{}
Let us first exclude some degenerate cases.
\fi{}
If~$k < 0$, then~$\ell < \dist(s, t)$, and we have a trivial \no-instance at hand.
If~$k=0$, then any solution must be a shortest path.
As any shortest walk is also a shortest path, we can use our algorithm for \fw{}, see \cref{thm:dft-fpt-r}.
Finally, we may assume that each of the~$n$ vertices in~$G$ reaches~$t$.
In all, we have that~$\dist(s, t) < \ell < n$ and thus~$0 < k < n-\dist(s, t)$.

\newcommand{\thmdfpabove}{%
	\label{thm:dfp-above}
	\fp{} can be solved in $r^{\bigO(r+k)} \cdot \ell n^2 m$ time,
	where~$n$ and $m$ are the number of vertices and arcs of the input digraph,
	and~$\ell$ is the length and~$k$ is the detour length of the desired path.
}
\ifshort
\begin{theorem}[\appref{thm:dfp-above}]
	\thmdfpabove
\end{theorem}
\fi

Our approach for \cref{thm:dfp-above} is to merge
our above techniques to keep track of the last~$r$ vertices
with a central observation for paths with detour length~$k$.
To this end, we will show that any hypothetical solution $P^*$ visits in regular intervals so-called \emph{distance separators} --- see \cref{fig:distsep} for an illustration.
At these points, 
we can partition the search space as we know that the subpath of~$P^*$ between two consecutive distance separators lies disjoint from any subpath between two other consecutive distance separators.
We then use a subroutine to compute a representative of all \rppref{} \abpaths{u}{v} of some fixed length whose running time is fixed-parameter tractable with respect to its length.
This fits into the promised running time
as any two distance separators~$u$ and~$v$ are at most~$2k+1$ vertices apart (\cref{thm:dist-sep}).
Indeed, such distance separators can be found in \emph{any} path with bounded detour length.
As mentioned earlier, this approach is inspired by works on parameterizations with respect to the detour length by \citet{bezakova2019detours} and \citet{zschoche2022restless}.
The challenge in our setting is that we need to keep track of the ordered representatives.

We start off with a basic observation.
We denote for every~$v \in V(G)$ by~$d(v) \coloneqq \dist(v, t)$ the distance to~$t$.

\newcommand{\thmvertexdistance}{%
	\label{thm:vertex-distance}
	For any \stpath{} $P = (s = v_0, \dots, v_\ell = t)$ with~$\ell \le d(s)+k$
	we have $i \le d(s) - d(v_i) + k$ for each~$i \in \nullto{\ell}$.
}
\ifshort
\begin{observation}[\appref{thm:vertex-distance}]
	\thmvertexdistance
\end{observation}
\fi
\appendixproof{thm:vertex-distance}
{
\begin{proof}
	Suppose that~$i > d(s) - d(v_i) + k$.
	As any \abpath{v_i}{t} has length at least~$d(v_i)$,
	the length of $P$ is at least~$i + d(v_i) > d(s) + k \ge \ell$.
\end{proof}
}

\begin{definition}[Distance separator]
	Let $P = (s = v_0, v_1, \dots, v_{\ell} = t)$ be a path with detour length~$k \coloneqq \ell - \dist(s, t)$.
	Then~$v_i$ is a \emph{distance separator} if
	\ifshort{}
	 $d(v_i) < d(v_j)$ for all $j < i$ and $d(v_i) > d(v_j)$ for all $j > i$.
	\fi{}
\end{definition}

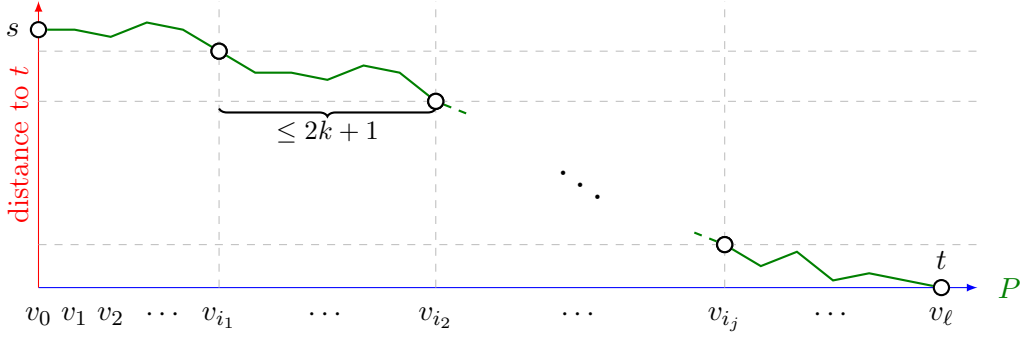
\begin{figure}[t]
	\centering
	\begin{tikzpicture}[scale=0.95]
	 \def\xr{0.5}
	 \def\yr{0.1} %
	 \def\xmx{26}
	 \def\ymx{40}

		\draw[->,>=latex,color=blue] (0,0) to (\xmx*\xr,0);
		\node at (\xmx*\xr,0)[label=0:{\textcolor{green!50!black}{$P$}}]{};
		\draw[->,>=latex,color=red] (0,0) to node[midway,above,sloped]{distance to~$t$}(0,\ymx*\yr);
		
		\node at (0,0)[label=-90:{$v_0$}]{};
		\node at (1*\xr,0)[label=-90:{$v_1$}]{};
		\node at (2*\xr,0)[label=-90:{$v_2$}]{};
		\node at (3.5*\xr,0)[label=-90:{$\cdots$}]{};
		\node at (8*\xr,0)[label=-90:{$\cdots$}]{};
		\node at (15*\xr,0)[label=-90:{$\cdots$}]{};
		\node at (22*\xr,0)[label=-90:{$\cdots$}]{};
		\node at (25*\xr,0)[label=-90:{$v_\ell$}]{};

		\node (s) at (0,36*\yr)[vertex,label=180:{$s$}]{};
		\node (z) at (25*\xr,0*\yr)[vertex,label=90:{$t$}]{};

		\tikzstyle{mypath}=[-,thick,green!50!black]
		\draw[mypath] (s) 
			to ++(1*\xr,0*\yr)
			to ++(1*\xr,-1*\yr)
			to ++(1*\xr,2*\yr)
			to ++(1*\xr,-1*\yr)
			to ++(1*\xr,-3*\yr) node[vertex] (a1){};
			;
		
		\draw[mypath] (a1) 
			to ++(1*\xr,-3*\yr)
			to ++(1*\xr,0*\yr)
			to ++(1*\xr,-1*\yr)
			to ++(1*\xr,2*\yr)
			to ++(1*\xr,-1*\yr)
			to ++(1*\xr,-4*\yr) node[vertex] (a2){}
			;
			
		\draw[mypath,dashed] (a2) to++ (1*\xr,-2*\yr); 
			
		\node (aL) at (19*\xr,6*\yr)[vertex]{};
		\draw[mypath,dashed] (aL) to++ (-1*\xr,2*\yr); 
		
		\node at ($(a2)!0.5!(aL)$)[scale=1.5]{$\ddots$};
		
		\draw[mypath] (aL) 
			to ++(1*\xr,-3*\yr)
			to ++(1*\xr,2*\yr)
			to ++(1*\xr,-4*\yr)
			to ++(1*\xr,1*\yr)
			to ++(1*\xr,-1*\yr)
			to (z)
			;
		
		\newcommand{\crossit}[1]{
			\draw[lightgray,dashed,thin] (0,0|-#1) to (\xmx*\xr,0|-#1);
			\draw[lightgray,dashed,thin] (#1|-0,0) to (#1|-0,\ymx*\yr);
			\node at (#1)[vertex]{};
		}
		\crossit{a1}
		\crossit{a2}
		\crossit{aL}
		
		\node at (a1|-0,0)[label=-90:{$v_{i_1}$}]{};
		\node at (a2|-0,0)[label=-90:{$v_{i_2}$}]{};
		\node at (aL|-0,0)[label=-90:{$v_{i_j}$}]{};

		\draw[decorate, decoration = {brace, mirror, amplitude=3pt}, thick] ([yshift=-1*\yr cm] a1.center |- a2.center) -- ([yshift=-1*\yr cm] a2.center) node [midway,yshift=-0.3cm,font=\small] {$\le 2k+1$};
	\end{tikzpicture}
	\caption{
		An exemplary \stpath{} $P$, circles marking distance separators.
		The~$x$-axis shows the vertices of~$P$ in the order of their appearance.
		The $y$-axis shows the distance of the current vertex to~$t$.
		Our algorithm exploits the property that the subpaths between any two distance separators are short (i.e., of length at most~$2k+1$) and internally vertex-disjoint.
	}
	\label{fig:distsep}
\end{figure}

By definition, if we have two distance separators $v_i$ and $v_j$, $j > i$, then we know that between~$v_i$ and~$v_j$, $P$ only visits vertices~$w$ with~$d(v_i) > d(w) > d(v_j)$.
\citet{zschoche2022restless} showed that a path with detour length~$k$ regularly visits distance separators.
We need a faintly different statement.

\newcommand{\thmdistsep}{%
	\label{thm:dist-sep}
	Let~$P = (s = v_0, v_1, \dots, v_p = v)$ be a path of length at most~$d(s) - d(v) + k$ and let~$v$ be a distance separator.
	Then, for all~$i \in \nullto{p-2k}$, there is a~$j \in \nullto{2k}$ such that~$v_{i+j}$ is a distance separator.
}
\ifshort
\begin{lemma}[\appref{thm:dist-sep}]
	\thmdistsep
\end{lemma}
\fi
\appendixproof{thm:dist-sep}
{
\begin{proof}
	Suppose that there is some~$i < p-2k$ such that none of the~$2k+1$ vertices~$v_i, \dots, v_{i+2k}$ is a distance separator.
	Let~$i$ be the largest such index, that is, $v_{i+2k+1}$ is a distance separator.
	We claim that
	for every~$x \in \xtoy{d(v_{i+2k})}{d(v_i)}$,
	there are two distinct~$y, z \in \xtoy{i}{i+2k}$ such that~$d(v_y) = d(v_z) = x$:
	If there is an~$x$ such that there is only one one vertex~$v_y$ with~$d(v_y) = x$ and~$v_y$ is not a distance separator,
	there exists a~$z \in \nullto{p}$ such that~$d(v_y)=d(v_z)$.
	As~$v_{i+2k+1}$ is a distance separator, we have~$z < i$.
	But then, as the path~$(s = v_0, \dots, v_z)$ has length at least~$d(s)-d(v_z)$ and the path~$(v_{i+2k}, \dots, v_p = v)$ has length at least~$d(v_{i+2k})-d(v)$,
	and as~$i+2k - z \ge 2k+1$,
	we have
	\[
		p \ge d(s) - d(v_z) + 2k+1 + d(v_{i+2k}) - d(v) \ge p + k+1,
	\]
	using~$p \le d(s) - d(v) + k$ --- a contradiction.
	Thus, our claim holds.

	From our claim we obtain by pigeonhole principle that~$\abs{\xtoy{i}{i+2k}} \ge 2 \abs{\xtoy{d(v_{i+2k})}{d(v_i)}}$.
	This yields
	\[
		d(v_i)- d(v_{i+2k}) = \abs{\xtoy{d(v_{i+2k})}{d(v_i)}} - 1 \le \textstyle\frac{1}{2}\abs{\xtoy{i}{i+2k}} - 1 = k - \frac{1}{2} < k.
	\]
	But as the length of the path~$(s = v_0, \dots, v_i)$ is at least~$d(s) - d(v_i)$ and the length of the path~$(v_{i+2k}, \dots, v_p = v)$ is at least~$d(v_{i+2k}) - d(v)$,
	we have
	\[
		p \ge (d(s)-d(v_i)) + 2k + (d(v_{i+2k}) - d(v)) \ge p-k + 2k + d(v_{i+2k}) - d(v_i),
	\]
	using~$p \le d(s) - d(v) + k$.
	That is, $d(v_{i}) - d(v_{i+2k}) \ge k$ --- a contradiction.
\end{proof}
}

Our algorithm can now guess the positions of the distance separators.
As the subpaths between the distance separators are (internally) vertex-disjoint,
we then only need to find an \rppref{} path that matches the color sequence of the subpath to the last distance separator and only uses vertices after this distance separator.
For any two distance separators~$u$ and~$v$ in our graph~$G$, we define
\begin{equation*}
	G_{u, v} \coloneqq G[B_{u,v}\cup\{u,v\}]
	\text{ and }
	B_{u, v} \coloneqq \begin{cases}
		\{w \in V(G) \mid d(u) > d(w) > d(v)\} & \text{ if } u \ne s,\\
		\{w \in V(G) \mid d(w) > d(v)\} & \text{ if } u = s.
	\end{cases}
\end{equation*}
On these graphs, we will compute an \orepresentative{} for the family of \rppref{} \abpaths{u}{v} of some length in~$G_{u,v}$.
Indeed, as we will append these paths to some \rppref{} \abpath{s}{u} that ends on some color sequence~$\tau$, we need the family to be $r$-compatible with~$\tau$.
We say that a path~$P = (v_0, \dots, v_q)$ \emph{fits} $\tau$ if~$\tau$ is $r$-compatible to~$(\col(v_0), \dots, \col(v_{\min\{r-1,q\}}))$.
We will need to compute an \orepresentative{} for the following family for any two distance separators~$u$ and~$v$ in our graph~$G$, any integer~$q \in \Nzero$, and any color sequence~$\tau$ of length at most~$r$:
\[
	\PPP^q_{\tau}(G_{u,v}) \coloneqq \left\{ \sigma\ \middle|\ 
	\begin{aligned}
		& \abs{\sigma} = \min \{q+1, r\} \text{ and there is an \rppref{} length-$q$ \abpath{u}{v}} \\
		& \text{in~$G_{u, v}$ that fits } \tau \text{ and whose color sequence ends on } \sigma
	\end{aligned}
	\right\}
\]

\ifshort{}
Computing such families can be done with an adaptation of \cref{alg:fpt-r} for walks,
the difference being that we additionally need to remember the set of vertices visited so far by our path%
\ifappendix{}
	(see \cref{thm:dfp-fpt-l} in the appendix).
\else{}
	(refer to the full version).
\fi{}%
Remembering these vertices comes at the cost of an additional running time factor of~$r^{\bigO(q)}$, where~$q$ is the length of the path.
As the path length is an upper bound for~$r$, this proves \fp{} to be fixed-parameter tractable with respect to the path length.
\fi{}

\toappendix{
\ifshort{}
\subsection{Fixed-Parameter Tractability by Length}
\label{app:sec:fpt-l}
\fi{}

We will next prove the following.

\begin{proposition}
	\label{thm:dfp-fpt-l}
	One can decide in~$r^{\bigO(r+\ell)} \cdot m$ time whether a given vertex-colored, $m$-arc digraph contains an \rppref{} \stpath{}.
\end{proposition}

Compared with the family $\WWW^p_v$ computed in \cref{alg:fpt-r} for walks, we additionally remember the set of vertices visited so far by our path.
As we will use unordered representatives to ensure that every vertex is visited at most once,
we will directly remember~$\pi(\sigma)$ 
instead of~$\sigma$.
As~$\pi(\sigma) \subseteq \colors \times \oneto{r}$, the universe of our family will be~$V(G) \cup (\colors \times \oneto{r})$.
However, our family will only contain sets~$X \uplus S$ such that~$X \subseteq V(G)$ and there exists a sequence~$\sigma$ such that~$\pi(\sigma) = S$.

Formally, we will compute a representative of the following for every~$p \in \nullto{\ell}$ and~$v \in V(G)$.

\begin{equation}
	\label{eq:path-table}
	\PPP^{p}_{v} \coloneqq \left\{
	\begin{gathered}
		X \uplus S\\
		\subseteq V(G) \cup (\colors \times \oneto{r})
	\end{gathered}
	\ \middle\lvert\ 
	\begin{aligned}
		& X \subseteq V(G),  \abs{X}=p+1, \text{ and}\\
		& \text{there is } \sigma=(a_1, \dots, a_{p'}) \text{ with}\\
		& p' = \min\{p+1, r\}, \pi(\sigma) = S,\\
		& \text{such that $G$ contains an \rppref{}}\\
		& \text{\abpath{s}{v} $P$ with } V(P)=X \text{ whose}\\
		& \text{color sequence ends on } \sigma
	\end{aligned}
	\right\}.
\end{equation}

Our algorithm works mostly in analogy to \cref{alg:fpt-r}; however it runs only on so-called \emph{partial representatives},
which are a relaxation of standard representatives in the sense that one does not need a representative for every set~$Y$ of some size, but only for sets that are contained in a specified family~$\YYY$.

\begin{definition}[Partial representative]
	\label{def:prepr}
	Let~$\TTT$ be a family of~$p$-element sets and let~$\YYY$ be a set family.
	A subfamily~$\hatty \TTT$ of~$\TTT$ is a \prepresentative{$\YYY$} (written $\hatty \TTT \prepr^\YYY \TTT$) if the following holds for every~$Y \in \YYY$:
	If~$\TTT$ contains a set~$X$ disjoint from~$Y$,
	then~$\hatty\TTT$ contains a set~$\hatty X$ disjoint from~$Y$.
\end{definition}

Note that if~$\YYY$ contains all sets of size at most~$q$, then the above definition coincides with the definition of (unordered) $q$-representatives.

Partial representatives still have the same transitivity property as unordered representatives.

\begin{observation}
	\label{obs:prepr:transitive}
	If~$\hatty\TTT \prepr^\YYY \tildy\TTT$ and~$\tildy\TTT \prepr^\YYY \TTT$, then~$\hatty\TTT \prepr^\YYY \TTT$.
\end{observation}

\begin{proof}
	Let~$Y \in \YYY$
	and let~$X \in \TTT$ with~$X \cap Y = \emptyset$.
	Then by definition of partial representatives,
	there is
	a~$\tildy X \in \tildy\TTT$ with~$\tildy X \cap Y = \emptyset$.
	Then there also exists
	a~$\hatty X \in \hatty\TTT$ with~$\hatty X \cap Y = \emptyset$, and the claim follows.
\end{proof}

Further, if all sets in~$\YYY$ have size at most~$q$, then any ``proper'' unordered \representative[$q$] is also a \prepresentative{$\YYY$}.

\begin{observation}
	\label{obs:prepr:repr}
	Let~$\TTT$ be a family of~$p$-element sets and let~$\YYY$ be a family of sets of size at most~$q$.
	If~$\hatty\FFF$ is a \representative[$q$] of~$\TTT$,
	then~$\hatty\FFF$ is a \prepresentative{$\YYY$} of~$\TTT$.
\end{observation}

\begin{proof}
	Let~$Y \in \YYY$ and let~$X \in \TTT$ be disjoint from~$Y$.
	Then $\hatty\FFF$ contains a set~$\hatty X$ that is also disjoint from~$Y$ as~$\abs{Y} \le q$, and the claim follows.
\end{proof}

Recall that we are given a graph~$G$ with two terminals~$s$ and~$t$, a coloring~$\col \colon V(G) \to \colors$, and two integers~$r$ and~$\ell$ as input.
In the following, let~$\YYY^q_r$ contain all those sets that can be partitioned into a subset of~$V(G)$ of size at most~$q$ and a subset of~$(\colors \times \oneto{r})$ of size at most~$r$ which is~$\pi'(\rho)$ for some color sequence~$\rho$, that is,
\[
	\YYY^q_r \coloneqq \{ Y \uplus T \mid \abs{Y} \le q,\, \abs{T} \le r, \text{ and there is a sequence } \rho \text{ with } \pi'(\rho) = T \}.
\]
Note that the sets in~$\YYY^q_r$ have size at most~$q+r$.
Our algorithm now works as follows.

\begin{alg}\label{alg:fpt-l}
	Set~$\hatty\PPP^0_{s} \ceq \{ (\col(s)) \}$
	and for all~$v \in V(G) \setminus \{s\}$, set~$\hatty\PPP^0_{v} \ceq \emptyset$.
	Now, for each~$p = 1, 2, \dots, \ell$,
	for each~$v \in V(G)$,
	compute~$\QQQ^p_v$, which is
	\begin{align*}
			&\displaystyle\bigcup_{u \in N^-(v)}
			\left\{
				\begin{gathered}
					(X \cup \{v\}) \cup {}\\
					\pi((a_1, \dots, a_{p+1}))
				\end{gathered}
				\ \middle\lvert\
				\begin{aligned}
					& v \notin X, (X\cup \pi((a_1, \dots, a_p))) \in \hatty\PPP^{p-1}_u,\\
					& \col(v) = a_{p+1} \notin (a_1, \dots, a_p)
				\end{aligned}
			\right\}\\
			\intertext{ if $p < r$ and}
			&\displaystyle\bigcup_{u \in N^-(v)}
			\left\{
				\begin{gathered}
					(X \cup \{v\}) \cup {}\\
					\pi((a_2, \dots, a_{r+1}))
				\end{gathered}
				\ \middle\lvert\
				\begin{aligned}
					& v \notin X, (X\cup \pi((a_1, \dots, a_r))) \in \hatty\PPP^{p-1}_u,\\
					& \col(v) = a_{r+1} \notin (a_1, \dots, a_r)
				\end{aligned}
			\right\}
	\end{align*}
	if~$p \ge r$, and compute an (unordered) \representative[$(\ell-p+r)$] $\hatty\PPP^p_v$ for~$\QQQ^p_{v}$ using \cref{thm:unordered-rep}.
\end{alg}

We next prove that~$\hatty\PPP^p_{v}$ are indeed partial representatives of~$\PPP^p_v$.
The proof is very similar to the one for \cref{lem:fpt-r-repr}.

\begin{lemma}
	\label{lem:fpt-l-repr}
	For each~$v \in V(G)$ and~$p \in \nullto{\ell}$,
	the family $\hatty \PPP^p_{v}$
	computed by \cref{alg:fpt-l}
	contains at most~$\binom{r'+r+\ell}{r+\ell-p}$ sets,
	where~$r' = r(r+1)/2$,
	and
	is a \prepresentative{$\YYY^{\ell-p}_r$} of $\PPP^p_{v}$ as defined in~\eqref{eq:path-table}.
\end{lemma}
\begin{proof}
	Our proof is by induction.
	By the initial assignments of $\hatty \PPP^0_{v}$ for each~$v \in V(G)$, the statement is correct for~$p=0$.
	Now, fix some~$p \in \oneto{\ell}$ and assume that~$\hatty\PPP^{p-1}_u$
	is a \prepresentative{$\YYY^{\ell-p+1}_r$} for~$\PPP^{p-1}_u$ for all~$u \in V(G)$.

	Let~$(Y \uplus \pi'(\rho)) \in \YYY^{\ell-p}_r$ and let~$\rho = (b_1, \dots, b_q)$.
	Suppose there exists a set~$(X \uplus \pi(\sigma)) \in \PPP^p_v$ that is disjoint from
	$(Y \uplus \pi'(\rho))$,
	that is,
	$X \cap Y = \emptyset$ and~$\sigma$ is $r$-compatible to~$\rho$.
	We claim that there exists~$(\hatty X \uplus \pi(\hatty\sigma)) \in \hatty\PPP^p_v$
	which is disjoint from~$(Y \uplus \pi'(\rho))$,
	thus proving that
	$\hatty\PPP^p_v$ is a \prepresentative{$\YYY^{\ell-p}_r$} of $\PPP^p_{v}$.
	As the bound on~$\abs{\hatty\PPP^p_v}$ then follows from \cref{thm:unordered-rep}, we are done once the claim is proven.
	We will prove the claim first for~$p \ge r$ and afterwards for~$p < r$.

	If~$p \ge r$, then~$\sigma$ is an $r$-sequence~$(a_2, \dots, a_{r+1})$,
	and there exists an \rppref{} length-$p$ \abpath{s}{v} $P$ whose color sequence ends on~$\sigma$ and whose vertex set is~$X$.
	Let~$a_1$ be the color that~$P$ visits just before visiting the colors in~$\sigma$, that is, the color sequence of~$P$ ends on~$(a_1, a_2,\dots, a_{r+1})$.
	Further, let~$u$ be the penultimate vertex visited by~$P$
	and let~$P'$ be the length-$(p-1)$ \subp{} of~$P$ ending on~$u$.
	Then~$P'$ ends on the sequence~$\sigma' \coloneqq (a_1, \dots, a_r)$ and~$V(P') = X \setminus \{v\} \eqqcolon X'$.
	Let~$\rho' \coloneqq (a_{r+1}) \circ \rho$ and let~$Y' \coloneqq Y \cup \{v\}$ (note that~$v \notin Y$).
	Observe that~$\sigma'$ is $r$-compatible to~$\rho'$,
	due to~$\sigma$ being $r$-compatible to~$\rho$ and~$P'$ being \rppref{}.
	Further, $X' \cap Y' = \emptyset$.
	Now, as $(Y' \uplus \pi'(\rho')) \in \YYY^{\ell-p+1}_r$,
	by our induction hypothesis and the definition of partial representatives,
	there exists~$(\hatty X' \uplus \pi(\hatty\sigma')) \in \hatty\PPP^{p-1}_u$
	which is disjoint from $(Y' \uplus \pi'(\rho'))$; thus
	$\hatty X' \cap Y' = \emptyset$ and $\hatty\sigma'$ is $r$-compatible to~$\rho'$.
	Let~$\hatty P'$ be the \abpath{s}{u} with~$V(\hatty P') = \hatty X'$ whose color sequence ends on~$\hatty\sigma'$ and let~$\hatty\sigma' \coloneqq (\hatty a'_1, \dots, \hatty a'_r)$.
	Define~$\hatty\sigma \coloneqq (\hatty a'_2, \dots, \hatty a'_r) \circ (a_{r+1})$ and~$\hatty X \coloneqq \hatty X' \cup \{v\}$.
	As~$u \in N^-(v)$, $v \notin \hatty X'$ and~$\col(v) = a_{r+1} \notin \hatty\sigma'$, we have that~$(\hatty X \uplus \pi(\hatty\sigma)) \in \QQQ^p_v$.
	Finally, observe that~$\hatty\sigma$ is $r$-compatible with~$\rho$ and that~$\hatty X \cap Y = \emptyset$;
	thus~$(\hatty X \uplus \pi(\hatty\sigma)) \cap (Y \uplus \pi'(\rho)) = \emptyset$,
	and we have that
	$\QQQ^p_v$ is a \prepresentative{$\YYY^{\ell-p}_r$} for~$\PPP^p_v$.
	Since $\hatty\PPP^p_v \repr^{(r+\ell-p)} \QQQ^p_v$, 
	we also have that~$\hatty\PPP^p_v$ is a \prepresentative{$\YYY^{\ell-p}_r$} for~$\QQQ^p_v$ by \cref{obs:prepr:repr} (note that the sets in $\YYY^{\ell-p}_r$ have size~$\ell-p+r$).
	Finally, due to the transitivity of partial representatives (\cref{obs:prepr:transitive}), the claim follows for~$p \ge r$.

	If~$p < r$, then~$\sigma$ is a~$(p+1)$-sequence~$(a_1, \dots, a_{p+1})$
	and there exists an \rppref{} length-$p$ \abpath{s}{v} $P$ whose color sequence ends on~$\sigma$.
	Indeed, $\sigma$ is the entire color sequence of~$P$.
	This case is similar to the above case, but there exists no color that is visited before~$\sigma$ in~$P$.
	Hence, in this case, $\sigma' \coloneqq (a_1, \dots, a_p)$
	and~$\hatty\sigma \coloneqq \hatty\sigma \circ (a_{p+1})$.
	The remainder of the proof is the same.
\end{proof}

\begin{lemma}
	\label{lem:fpt-l-time}
	\cref{alg:fpt-l} runs in $r^{\bigO(\ell+r)} \cdot m$ time, where~$m$ is the number of arcs in the input digraph.
\end{lemma}
\begin{proof}
	Let~$v \in V(G)$ and~$p \in \oneto{\ell}$.
	Each set in~$\QQQ^p_v$ has size at most~$r'+p$, where~$r' = r(r+1)/2$.
	To compute~$\QQQ^p_v$, we need to consider the sets~$\hatty\PPP^{p-1}_u$ for each in-neighbor~$u$ of~$v$.
	By \cref{alg:fpt-l}, $\hatty\PPP^{p-1}_u$ is an $(r+\ell-p+1)$-representative for a family of sets of size~$r'+p-1$.
	Hence, by \cref{thm:unordered-rep}, $\hatty\PPP^{p-1}_u$ contains~$\bigO\bigl(\binom{r'+r+\ell}{r'+p}\bigr)$ sets.
	As~$r'+r \le r^2$ whenever~$r \ge 3$,
	we have
	\begin{align*}
		\binom{r'+r+\ell}{r'+p} &= \binom{r'+r+\ell}{r+\ell-p} \le \binom{r^2+\ell}{r+\ell-p}
		= \sum_{k=0}^{r+\ell-p} \textstyle\binom{r^2}{k} \binom{\ell}{r+\ell-p -k}\\
		\textstyle &\le 2^\ell \sum_{k=0}^{r+\ell-p} \textstyle\binom{r^2}{k}
		\textstyle \le 2^\ell (r+\ell-p) r^{2(r+\ell-p)} \in r^{\bigO(\ell)},
	\end{align*}
	wherein the last equality in the first line holds due to Vandermonde's convolution~\cite[Eq.~5.22]{graham1994concrete}.
	Thus, $\abs{\QQQ^p_v} \le \deg(v) \cdot r^{\bigO(\ell)}$, and the time to compute~$\QQQ^p_v$ is linear in its size.
	Thus, the time to compute the~$(r+\ell-p)$-representative~$\hatty\PPP^p_v$ for~$\QQQ^p_v$ is upper-bounded by
	\begin{align*}
		& \abs{\QQQ^p_v} \cdot \textstyle \binom{r'+r+\ell}{r'+p} (r'+p)^\omega + \abs{\QQQ^p_v} \cdot \binom{r'+r+\ell}{r'+p}^{\omega-1}
		\in \abs{\QQQ^p_v} r^{\bigO(\ell)} \subseteq r^{\bigO(\ell)}.
	\end{align*}
	Doing this for each~$v \in V(G)$ and each~$p \in \oneto{\ell}$ yields the claimed running time by the handshaking lemma.
\end{proof}

\cref{thm:dfp-fpt-l} now trivially follows from \cref{lem:fpt-l-repr,lem:fpt-l-time}.

\begin{proof}[Proof of \cref{thm:dfp-fpt-l}]
	By \cref{lem:fpt-l-repr}, \cref{alg:fpt-l} correctly computes a \prepresentative{$\YYY^{\ell-p}_r$} $\hatty \PPP^p_{t}$ for the family~$\PPP^p_{t}$ in~\eqref{eq:path-table} for each~$p \in \nullto{\ell}$,
	that is, there exists an \rppref{} \stpath{} of length~$\ell$ if and only if~$\hatty \PPP^\ell_t \ne \emptyset$.
	The running time follows from \cref{lem:fpt-l-time}.
\end{proof}
} %

\newcommand{\thmbruteforce}{%
	\label{thm:brute-force}
	Given a digraph~$G$ with~$m$ arcs, an integer~$r$, two distance separators~$u, v \in V(G)$, an integer~$q \in \Nzero$, and a color sequence~$\tau$ of length at most~$r$, one can compute in~$r^{\bigO(r+q)} \cdot m$ time an \orepresentative{} for~$\PPP^q_{\tau} (G_{u,v})$ of size at most~$r^{\bigO(r+q)}$.
}
\ifshort
\begin{lemma}[\appref{thm:brute-force}]
	\thmbruteforce
\end{lemma}
\fi
\appendixproof{thm:brute-force}{
\begin{proof}
	Let~$\tau = (a_1, \dots, a_p)$ be a color sequence.
	We will create an auxiliary digraph~$G^*$ which consists of a copy of~$G_{u,v}$ and a path~$(s = v'_1, \dots, v'_p, u)$
	where for each~$i \in \oneto{p}$, $v_i$~has color~$a_i$.
	If~$p=0$, then we identify~$u$ with $s$.
	Observe that any \rppref{} \abpath{s}{v} $P$ of length~$p+q$ in~$G^*$ first visits the vertices with color sequence~$\tau$ and then visits the vertices of an \abpath{u}{v} $P'$ of length~$q$.
	Clearly, $P'$ fits~$\tau$.

	We now run
	\cref{alg:fpt-l} with input graph~$G^*$, length~$p+q$, and terminals~$s$ and~$v$ to compute~$\hatty\PPP^q_v$ for~$G^*$.
	Then, by \cref{lem:fpt-l-repr,thm:pi}, the family $\hatty\PPP^q_{\tau}(G_{u,v}) \coloneqq \{ \sigma^* \mid X \cup \pi(\sigma) \in \hatty\PPP^{p+q}_v \}$ is an \orepresentative{} for~$\PPP^q_{\tau}(G_{u,v})$.
	By \cref{lem:fpt-l-time,lem:fpt-l-repr}, the running time and size bounds are met.
\end{proof}
} %

Now that we know how to compute the families $\hatty\PPP^q_{\tau}(G_{u,v})$, we can state the main algorithm.
Herein, for every~$p \in \nullto{\ell}$ and~$v \in V(G)$, we are interested in the family
\begin{equation}
	\label{eq:table-above}
	\RRR^{p}_{v} \coloneqq \left\{ \sigma\ \Bigg\lvert\ 
	\begin{aligned}
		& \abs{\sigma} = \min \{p+1, r\} \text{ and there is an \rppref{} length-$p$} \\
		& \text{\abpath{s}{v} in~$G_{s, v}$ whose color sequence ends on } \sigma
	\end{aligned}
	\right\}.
\end{equation}
Hence, there is a length-$\ell$ \rppref{} \stpath{} $P^*$ if and only if~$\RRR^\ell_t$ is nonempty.
As, by \cref{thm:vertex-distance}, $P^*$ will have~$v$ as its~$p$-th vertex only if~$p \le d(s) - d(v) + k$, we only need to consider those families~$\RRR^p_v$ for which this inequality holds.

\begin{alg}
	\label{alg:above}
	Set~$\hatty\RRR^0_s \coloneqq \{(\col(s))\}$ and for all~$v \in V(G) \setminus \{s\}$, set~$\hatty\RRR^0_v \coloneqq \emptyset$.
	Now, 
	for each~$p = 1, 2, \dots, \ell$, 
	for each~$v \in V(G)$ with~$p \le d(s) - d(v) + k$, 
	compute
	\begin{equation}
		\label{eq:above-recurrence}
		\SSS^p_v \coloneqq
		\bigcup_{\mathclap{\substack{u \in V(G),\\q \in \oneto{\min\{2k+1, p\}}}}}
		\hspace{0.60em}
			\left\{\sigma' \! \circ \sigma \left\vert\,
				\begin{aligned}
					& \exists \sigma'' : \sigma'' \!\! \circ \sigma'\! \in \hatty\RRR^{p-q}_u\!,\, \abs{\sigma'} = \max\{0, \min\{p\!-\!q\!+\!1, r\!-\!q\}\}, \\
					& \abs{\sigma}=\min\{q, r\}, \text{ and } \sigma \in \hatty\PPP^q_{(\sigma'' \circ \sigma')}(G_{u,v})
				\end{aligned}
				\right.
			\right\},
	\end{equation}
	and compute~$\hatty\RRR^p_v \orepr^r \SSS^p_v$ using \cref{thm:ordered-rep}.
	Return \yes{} if and only if~$\hatty\RRR^{\ell'}_t \ne \emptyset$ for some $\ell' \in \oneto{\ell}$.
\end{alg}

We prove the correctness by showing that indeed $\hatty\RRR^p_v$ is an \orepresentative{} for~$\RRR^p_v$.
The main part herein is to prove this property for~$\SSS^p_v$.
In the same step we also analyze the running time of the algorithm.
With following lemma, the proof of \cref{thm:dfp-above} is immediate.

\newcommand{\thmaboverepr}{%
	\label{thm:above-repr}
	For each~$v \in V(G)$ and~$p \in \nullto{\ell}$ with~$p \le d(s) - d(v) + k$,
	the family~$\hatty\RRR^p_{v}$ computed in \cref{alg:above}
	is of size at most~$(r \cdot e)^r$
	and is an \orepresentative{} for $\RRR^p_v$ as defined in \eqref{eq:table-above}.
	Moreover, \cref{alg:above} runs in~$r^{\bigO(r+k)} \cdot \ell n^2 m$ time.
}
\ifshort
\begin{lemma}[\appref{thm:above-repr}]
	\thmaboverepr
\end{lemma}
\fi
\appendixproof{thm:above-repr}
{
\begin{proof}
	Our proof is by induction.
	By the initial assignments of $\hatty \RRR^0_{v}$ in \cref{alg:fpt-r}, the statement is correct for~$p=0$.
	Now, fix some~$p \in \oneto{\ell}$ and assume that for every~$p' \in \nullto{p-1}$ and every~$u \in V(G)$ with~$p' \le d(s)-d(u)+k$,
	the family $\hatty\RRR^{p'}_u$ is an \orepresentative{} for $\RRR^{p'}_u$.

	Let~$\lambda$ be a sequence on~$\colors$ with~$\abs{\lambda} \le r$ and let~$v \in V(G)$ be a vertex with~$p \le d(s)-d(v)+k$.
	Suppose that there exists a sequence~$\rho \in \RRR^p_v$ that is $r$-compatible to~$\lambda$.
	That is, there exists an \rppref{} length-$p$ \abpath{s}{v} $P = (s = v_0, v_1, \dots, v_p = v)$ in~$G_{s,v}$ whose color sequence ends on~$\rho$.

	We claim that there exists a $\hatty\rho \in \SSS^p_v$ such that $\hatty\rho$ is $r$-compatible to~$\lambda$, and there is an \rppref{} length-$p$ \abpath{s}{v} $\hatty P$ in~$G_{s,v}$ whose color sequence ends on~$\hatty\rho$.
	Then, $\hatty \rho \in \RRR^p_v$, and~$\SSS^p_v$ is an \orepresentative{} for~$\RRR^p_v$.
	As~$\hatty\RRR^p_v \orepr^r \SSS^p_v$ by \cref{thm:ordered-rep}, we have~$\hatty\RRR^p_v \orepr^r \RRR^p_v$ by \cref{thm:transitive}.
	Further, the size bound follows from \cref{thm:ordered-rep}, and the lemma is proven.

	Let us prove our claim.
	Let~$q \in \oneto{2k+1}$ be the smallest number such that~$u \coloneqq v_{p-q}$ is a distance separator --- such a $q$ exists due to \cref{thm:dist-sep} and the fact that~$p \le d(s)-d(v)+k$.
	Then the subpath~$P' \coloneqq (s = v_0, \dots, v_{p-q} = u)$ of~$P$ is contained in~$G_{s, u}$,
	and the subpath~$P'' \coloneqq (u = v_{p-q}, \dots, v_{p})$ of~$P$ is contained in~$G_{u, v}$.

	Choose~$\sigma$ and~$\sigma'$ such that~$\abs{\sigma} = \min\{q, r\}$ and that~$\rho = \sigma' \circ \sigma$ (note that~$\sigma'$ may be empty).
	Further, choose~$\sigma''$ such that~$\abs{\sigma'' \circ \sigma'} = \min\{p-q+1,r\}$ and let~$\tau \coloneqq \sigma'' \circ \sigma'$.
	Then the color sequence of~$P'$ ends on~$\tau$.
	As~$P'$ is a subpath of~$P$, it is \rppref{}.
	Further, as its length is~$p-q$ and it is contained in~$G_{s, u}$, we have~$\tau \in \RRR^{p-q}_u$.
	We next define $\lambda'$ as follows.
	If~$r \le q$, then~$\lambda' \coloneqq (\col(v_{p-q+1}), \dots, \col(v_{p-q+r}))$ and note that~$\tau$ is $r$-compatible to~$\lambda'$.
	If~$r > q$, then~$\lambda' \coloneqq (\col(v_{p-q+1}), \dots, \col(v_{p})) \circ \lambda''$, where~$\lambda''$ consists of the first~$\min\{r-q, \abs{\lambda}\}$ entries of~$\lambda$.
	In this case, $\tau$ is also $r$-compatible to~$\lambda'$.
	Hence, by induction hypothesis,
	as~$\tau$ is $r$-compatible to~$\lambda'$,
	there exists a~$\hatty\tau \in \hatty\RRR^{p-q}_u$ that too is~$r$-compatible to~$\lambda'$.
	Let~$\hatty P' \coloneqq (\hatty v_0, \dots, \hatty v_{p-q})$ be the corresponding \abpath{s}{u} whose color sequence ends on~$\hatty\tau$.

	Recall that~$P''$ is contained in~$G_{u,v}$, fits~$\tau$, and ends on~$\sigma$, which is $r$-compatible to~$\lambda$.
	Hence, $\sigma \in \PPP^q_{\tau}(G_{u,v})$.
	Then, by \cref{thm:brute-force}, there exists a sequence~$\hatty\sigma \in \hatty\PPP^q_{\tau}(G_{u,v})$ that is $r$-compatible to~$\lambda$,
	and~$\hatty\sigma$ will be found by \cref{alg:above}.
	Let~$\hatty P'' \coloneqq (u = \hatty v_{p-q}, \dots, \hatty v_p = v)$ be the corresponding \abpath{u}{v} in~$G_{u,v}$.

	As~$\hatty P'$ is in~$G_{s,u}$ and~$\hatty P''$ is in~$G_{u,v}$, and the two graphs only have vertex~$u$ in common,
	$\hatty P \coloneqq (\hatty v_0, \dots, \hatty v_{p-q}, \hatty v_{p-q+1}, \dots, \hatty v_p)$ is a valid path in~$G_{s, v}$.
	Define~$\hatty\rho$ such that~$\abs{\rho} = \max\{r, p+1\}$ and the color sequence of~$\hatty P$ ends on~$\hatty\rho$.
	Observe that~$\hatty\tau \circ \hatty\sigma \circ \lambda$ is \rppref{}.
	As~$\hatty\rho$ is a substring at the end of~$\hatty\tau \circ \hatty\sigma$, we have that~$\hatty\rho \circ \lambda$ is \rppref{}.
	This implies that~$\hatty\rho$ is $r$-compatible to~$\lambda$.
	This proves our claim.

	Let us now focus on the running time of \cref{alg:above}.
	Let~$v \in V(G)$ and~$p \in \oneto{\ell}$.
	To compute~$\SSS^p_v$, we need to consider the sets~$\hatty\RRR^{p-q}_u$ for each~$u \in V(G)$ and each~$q \in \oneto{\min\{2k, p\}}$.
	By \cref{thm:above-repr}, $\abs{\hatty\RRR^{p-q}_u} \le (r\cdot e)^r$.
	For each~$\sigma'' \circ \sigma \in \hatty\RRR^{p-q}_u$ we then need to compute~$\hatty\PPP^q_{\sigma''\circ\sigma'}(G_{u,v})$, which, by \cref{thm:brute-force} takes~$r^{\bigO(r+q)} \cdot m$ time and, by \cref{lem:fpt-l-repr}, is of size at most~$\binom{r'+r+q}{r'+q} \in r^{\bigO(r+q)}$.
	Thus, $\SSS^p_v$ contains at most~$n \cdot 2k \cdot (r\cdot e)^r \cdot r^{\bigO(r+k)} \subseteq r^{\bigO(r+k)} \cdot n$ sets, and the running time to compute~$\SSS^p_v$ is at most
	\[
		n \cdot 2k \cdot (r\cdot e)^r \cdot r^{\bigO(r+k)} \cdot m = r^{\bigO(r+k)} \cdot nm.
	\]
	Finally, the time to compute~$\hatty\RRR^p_v \orepr^r \SSS^p_v$ is
	\[
		\bigO \big( \abs{\SSS^p_v} \cdot (r\cdot e)^r r^\omega + \abs{\SSS^p_v} \cdot (r\cdot e)^{(\omega-1)r} \big) \subseteq r^{\bigO(r+k)}. %
	\]
	Doing this for every~$v \in V(G)$ and~$p \in \oneto{\ell}$ results in the stated running time.
\end{proof}
}

\ifshort{}\fi{}
\appendixproof{thm:dfp-above}
{
\begin{proof}[Proof of \cref{thm:dfp-above}]
	If~$\ell = \dist(s, t)$, 
	then any \stwalk{} of length~$\ell$ is also a path; hence we can use \cref{thm:dft-fpt-r} to compute a solution within the claimed time.
	Otherwise, we use \cref{alg:above}.
	By \cref{thm:above-repr}, the algorithm runs in the claimed running time and computes for every~$v \in V(G)$ and each~$p \in \nullto{\ell}$ with~$p \le d(s)-d(v)+k$
	an \orepresentative{} of~$\RRR^{p}_v$.
	Clearly, if the algorithm returns \yes{}, then there exists an \rppref{} \stpath{} of length~$\ell$.
	Conversely, if there exists an \rppref{} \stpath{} $P^*$ of length~$\ell$,
	then by \cref{thm:vertex-distance}, $P^*$ will contain a vertex~$v \in V(G)$ as its~$p$-th vertex only if $p \le d(s)-d(v)+k$.
	More specifically, as~$\ell \le d(s)-d(t)+k = d(s)-0+k$, the algorithm correctly determines whether~$\hatty\RRR^{\ell'}_t \ne \emptyset$ and thus, whether~$\RRR^{\ell'}_t \ne \emptyset$ for every $\ell' \in \oneto{\ell}$.
\end{proof}
}

\section{Conclusion}
\label{ft:sec:conclusion}

We introduced a \emph{local rainbow} constraint to the classic problem of finding \stpaths{} and walks, 
modeling scenarios in which resources (i.e., colors) are replenished over time.
For walks, we are able to prove fixed-parameter tractability for the locality parameter $r$ thanks to a new adaptation of the representative sets technique.
In contrast, \fp{} remains \NP-hard even for constant~$r$ due to the added non-local constraint of forbidding self-intersections.
However, when the allowed length of the path is not too large in comparison to the distance between its endpoints, then the no-intersection constraints effectively become local again.
This is exploited to prove \fp{} to be fixed-parameter tractable with the combined parameter~$r + k$
where $k$~is the detour length.

Towards future work, we believe that \pprefness{} is only the tip of the iceberg when it comes to interesting local constraints.
A straightforward generalization would be to allow for multi-colored vertices, so as to model a setting in which multiple types of resources can be used at once.
Another natural variant would be to relax the \pprefness{} constraint of the subpaths, allowing some bounded number of vertices to share the same color.

One could also extend our \pprefness{} constraint to other connectivity problems.
Canonical candidates would be the traveling salesperson problem or the problem of finding multiple disjoint \stpaths{}.
Similarly, one could be interested in finding Steiner trees in which all subpaths are \ppref{} (a generalization of rainbow Steiner trees \cite{ferone2022rainbow}),
or vertex sets whose deletion destroys all \stpaths{} that are not \ppref{}.

We already observed that in practice, the locality constraint is usually motivated by some regeneration over time.
Therefore, it may be sensible to study the \pprefness{} constraints also on temporal graphs, such as finding temporal walks~\cite{himmel2020walks} or paths~\cite{casteigts2021paths,zschoche2022restless}. 

\printbibliography

\ifshort{}\ifappendix{}
\clearpage
\onecolumn
\appendix
\appendixProofText
\fi{}\fi{}

\end{document}